\documentclass[manuscript,screen, nonacm]{acmart}

\usepackage{lscape}
\usepackage{graphicx}
\usepackage{rotating}

\usepackage{lipsum}
\usepackage{algorithm}
\usepackage[noend]{algpseudocode}
\usepackage{tikz}
\usetikzlibrary{shadows.blur}
\usetikzlibrary{shapes.symbols}
\usetikzlibrary{positioning}
\usetikzlibrary{shapes.geometric}
\usetikzlibrary{shapes.misc}
\usepackage{hyperref}

\usepackage[normalem]{ulem}
\useunder{\uline}{\ul}{}
\usepackage{longtable}
\usepackage[utf8]{inputenc}
\usepackage{graphicx}
\usepackage{array}

\usepackage{pifont}
\newcommand{\cmark}{\ding{51}}%
\newcommand{\xmark}{\ding{55}}%

\DeclareMathOperator*{\argmax}{arg\,max}

\algdef{SE}{Upon}{EndUpon}[2]{\textbf{upon} \(\mbox{#1}\) \textbf{from}  \(\mbox{#2}\) \textbf{do}}{\textbf{end upon}}
\algtext*{EndUpon}

\AtBeginDocument{%
  \providecommand\BibTeX{{%
    \normalfont B\kern-0.5em{\scshape i\kern-0.25em b}\kern-0.8em\TeX}}}

\setcopyright{acmcopyright}
\copyrightyear{2018}
\acmYear{2018}
\acmDOI{XXXXXXX.XXXXXXX}

\acmConference[Conference acronym 'XX]{Make sure to enter the correct
  conference title from your rights confirmation emai}{June 03--05,
  2018}{Woodstock, NY}
\acmPrice{15.00}
\acmISBN{978-1-4503-XXXX-X/18/06}

\begin{document}

\title{Ethereum Proof-of-Stake and the Probabilistic Bouncing Attack}



\author{Ulysse Pavloff}
\orcid{0000-0003-4125-3306}
\affiliation{%
  \institution{Université Paris-Saclay, CEA, List}
  \city{Palaiseau} 
  \country{France} 
}

\author{Yackolley Amoussou-Guenou}
\orcid{0000-0002-6942-0412}
\affiliation{%
  \institution{Université Paris-Saclay, CEA, List}
  \city{Palaiseau} 
  \country{France} 
}

\author{Sara Tucci-Piergiovanni}
\orcid{0000-0001-9738-9021}
\affiliation{%
  \institution{Université Paris-Saclay, CEA, List}
  \city{Palaiseau} 
  \country{France} 
}

\renewcommand{\shortauthors}{Pavloff, Amoussou-Guenou, and Tucci-Piergiovanni.}

\begin{abstract}
Ethereum has undergone a recent change called \textit{the Merge}, which made Ethereum a Proof-of-Stake blockchain shifting closer to BFT consensus. Ethereum, which wished to keep the best of the two protocol designs (BFT and Nakomoto-style), now has a convoluted consensus protocol as its core. The result is a blockchain being possibly produced in a tree-like form while participants try to finalize blocks. We categorize different attacks jeopardizing the liveness of the protocol. The Ethereum community has responded by creating patches against some of them. We discovered a new attack on the patched protocol. To support our analysis, we propose a new high-level formalization of the properties of liveness and availability of the Ethereum blockchain, and we provide a pseudo-code. We believe this formalization to be helpful for other analyses as well.
Our results yield that the Ethereum Proof-of-Stake has safety but only probabilistic liveness. The probability of the liveness is influenced by the parameter describing the time frame allowed for validators to change their mind about the current main chain. 
\end{abstract}



\received{20 February 2007}
\received[revised]{12 March 2009}
\received[accepted]{5 June 2009}

\maketitle

\section{Introduction}
\label{sec:Introduction}
Ethereum has recently undergone a major change in its protocol, successfully passing from proof-of-work to proof-of-stake. The change underpins an entirely new consensus protocol, which brings Byzantine fault-tolerance to Ethereum. The main design goal is to keep using a Nakamoto-style consensus, i.e., a protocol that constantly creates blocks in a tree-like form and selects a branch as the current chain using a fork-chain rule. 
However, a mechanism (called finality gadget) incrementally finalizes blocks in the chain as opposed to pure Nakamoto-style consensus. A \textit{finalized block} is a block that is voted by at least two-thirds of validators\footnote{To become a validator, one needs to ``stake'' an amount of 32 ETH (the native cryptocurrency of the blockchain).}. In a system with  less than one-third of Byzantine validators, a finalized block is never revoked.  

Interestingly, this design aims at guaranteeing, at the same time, the availability of the chain (to let new blocks be continued added) and consistency, i.e.,  uniqueness of a finalized chain's prefix. Note that classical BFT consensus protocols re-adapted to blockchains such as Tendermint \cite{buchman_latest_2018}  for the Cosmos blockchain \cite{kwon_cosmos_2016} or Tenderbake \cite{astefanoaei_tenderbake_2021} for Tezos blockchain \cite{goodman_tezos_2014}, finalize one block at the time: for each height of the blockchain only one block is ever added. On the other hand, Ethereum Proof-of-Stake (PoS) builds a common prefix, but the suffix can change: for a given height of the blockchain, different blocks can be seen at that height over time. The advantage of this approach is to always make progress, regardless of Byzantine behavior and network partitions, while classical  BFT consensus protocols stop producing blocks during asynchronous periods and attacks. 
Ethereum PoS tries then to provide consistency without renouncing availability. This seems to be in striking contrast with the CAP theorem \cite{gilbert_brewer_2002}, which states that it is impossible to guarantee progress and consistency in the case of network partitions. 
The caveat here is that Ethereum PoS maintains a data structure where the prefix is consistent and finalized only when possible, while the suffix can grow without being consistent. The resulting protocol, however, is quite involved. For this reason, we aim to provide a formal ground for analysis of the Ethereum PoS protocol. 

We first provide a novel formalization of the Ethereum PoS blockchain properties, which are a combination of properties of Nakamoto-style blockchains and BFT-style ones. We then provide a high-level formalization of the protocol itself through pseudo-code. Formalization allows for analysis of the protocol in terms of its properties. 

In this paper, we study the security and the liveness of the protocol. Liveness is the ability to always finalize new blocks and security is the impossibility of having two checkpoints finalized on different chains. Our analysis reveals a new possible attack on the protocol's liveness. 
Indeed previous work has already pointed out the risks of a so-called \textit{bouncing attack} on liveness \cite{nakamura_analysis_2019}. A bouncing attack is an attack that prevents the chain from being finalized because the main chain selected through the fork choice rule continually bounces between two alternative branches. After the attack was identified, the Ethereum community responded by implementing a patch to the protocol to prevent this attack. The patch aims at mitigating bouncing by forcing validators to stick to a chain after a while. Interestingly, we managed to find a new bouncing attack on the patched version of the protocol. We found that the bouncing can be repeated over time but with decreasing probability of success. This shows that the liveness of the patched protocol is probabilistic. Note that the attack is plausible in a Byzantine environment since it only relies on the Byzantine validator's capacity to withhold votes and to release them at the right time to make honest validators change their mind on the chosen chain.

The paper is organised as follows: Section \ref{sec:systemModel} elaborates on the system model while Section \ref{sec:blockchainProperties} defines the properties essential to our formalization. In Section \ref{sec:ethereumConsensusProtocol}, we explain and formalize the Ethereum PoS protocol providing all the materials necessary for its understanding, including pseudo-code. Section \ref{sec:livenessAttack} presents a new liveness attack still possible in the current version of the Ethereum PoS protocol. Section \ref{sec:safety} proves the safety of Ethereum PoS. We present the related works in Section \ref{sec:relatedWorks}, and categorize the different types of attacks in Section \ref{sec:categorizationAttacks}. We conclude in Section \ref{sec:conclusion}.

\section{System model}
\label{sec:systemModel}

We consider a system composed of a finite set $\Pi$ of processes called \emph{validators}\footnote{At the implementation level, validators are the processes with ETH staked that allow them to vote as part of the consensus protocol.}. There are a total of $n$ validators.
Each validator has an associated public/private key pair for signing and can be identified by its public key.
We assume that digital signatures cannot be forged. 
Validators have synchronized clocks\footnote{Clocks can be offset by at most $\tau$,
this way, the offset can be captured as part of the network delay.}. 
Time is measured by periods of 12 seconds called \emph{slot}s, a period of 32 slots is called an \emph{epoch}.


\paragraph*{Network} Processes communicate by message passing.
We assume the existence of an underlying broadcast primitive, which is a best effort broadcast. This means that when a correct process broadcasts a value, all the correct processes eventually deliver it. Messages are created with a digital signature.

We assume a \emph{partially synchronous model} \cite{dwork_consensus_1988}, where after some unknown Global Stabilization Time (\texttt{GST}), the system becomes synchronous, and there is a finite known bound $\Delta$ on the message transfer delay. 
Note that even if we have synchronized clocks, having an asynchronous network before \texttt{GST} still makes the system partially synchronous.

\paragraph*{Fault Model}
Validators can be \emph{correct} or \emph{Byzantine}. Correct validators (also called honest validators) follow the protocol,
while Byzantines ones may arbitrarily deviate from the protocol\footnote{Since in this paper we are only interested in the consensus part of the protocol, we only characterize validator's behavior. For clients submitting transactions, as in any blockchain, we assume they can be Byzantine.}.
We denote by $f$ the number of Byzantine validators, with $f<n/3$. 

\section{Blockchain properties}
\label{sec:blockchainProperties}

In our analysis, we will continuously use the term Ethereum Proof-of-Stake as \cite{schwarz_three_2021} to name the new protocol of Ethereum\footnote{Other appellation such as Ethereum 2.0 or Consensus Layer can be found, we have chosen to stick with Ethereum Proof-of-Stake.}. To begin our analysis of Ethereum Proof-of-Stake, we start by defining the terms and properties we will investigate.

Similarly to \cite{anceaume_abstract_2018}, we formalize the blockchain data structure 
as a \emph{BlockTree}. 
Indeed the blockchain takes the form of a tree in which every node is a block pointing to its unique parent, and the tree's root is the \emph{genesis block}. 
Among the different branches of the BlockTree, the protocol indicates a unique branch, or chain, to build upon with a so-called fork choice rule (e.g., the longest chain rule in Bitcoin).
The selected chain is called the \emph{candidate chain}.



\begin{definition}[\textbf{Candidate chain}]
 We call \textbf{candidate chain} the chain designated as the one to build upon by the fork choice rule.
 Considering the view of the chain of an honest validator $i$, $i$'s associated candidate chain is noted $C_i$. 
\end{definition}

The blocks in the candidate chain can be finalized or not.

\begin{definition}[\textbf{Finalized block}]
A block is finalized for a validator $i$ if and only if  the block cannot be revoked, i.e., it permanently belongs to the candidate chain $C_i$ .
\end{definition}

\emph{Note:} It stems from the definition that all the predecessors of a finalized block are finalized.

\begin{definition}[\textbf{Finalized chain}]
The finalized chain is the chain constituted of all the finalized blocks.
\end{definition}

\emph{Note:} The finalized chain $C_{fi}$ is always a prefix of any candidate chain $C_i$.\\



To analyse the protocol, one needs to examine the capability of the Ethereum Proof-of-Stake protocol to construct a consistent blockchain (safety), to allow validators to add blocks despite network partitions and failures (availability), and to make progress on the finalization of new blocks (liveness). These are paramount properties characterizing blockchains.
Safety, availability, and liveness are expressed as follows:
\begin{definition}[\textbf{Safety}]
A blockchain is consistent or \textbf{safe}
if for any two correct validators with a finalized chain, then one chain  is necessarily the prefix of the other. More formally, for two validators $i$ and $j$ with respective finalized chain $C_{fi}$ and $C_{fj}$, then $C_{fi}$'s is the prefix of $C_{fj}$ or vice-versa.
\end{definition}


\begin{definition}[\textbf{Availability}]
A blockchain is \textbf{available} if the following two conditions hold: (1) any correct validator is able to append a block to its candidate chain in bounded time, regardless of the failures of other validators and the network partitions;  (2) the candidate chains of all correct validators are eventually growing, i.e., given a block $b_k$ added to a candidate chain at a distance $d$ from the genesis block $b_0$, where the distance is the number of blocks separating $b_k$ from $b_0$, then eventually a block $b_l$ will be added to the candidate chain at a distance $d'>d$. 
\end{definition}

\begin{definition}[\textbf{Liveness}]
A blockchain is \textbf{live} if the finalized chain is ever growing.
\end{definition}

The fundamental difference between the finalized and the candidate chain lies in the fact that blocks of the finalized chain can never be revoked, while the candidate chain can change from one branch to another in the tree so that a suffix of blocks of the previously selected branch might be revoked. 
 Availability, on the other hand, guarantees that adding blocks to the candidate chain is a wait-free operation whose time to complete does not depend on network failures or Byzantine behaviors. Availability also implies that blocks are constantly added in such as way that the height of the candidate chain eventually grows. This property avoids the pathological scenario in which all the blocks are added to the genesis block to form a star.  



As in any distributed system, blockchains are faced with the dilemma brought by the CAP-Theorem \cite{gilbert_brewer_2002}. This theorem states that no distributed system can satisfy these three properties at the same time:  \emph{consistency}, \emph{availability}, and \emph{partition tolerance}. Indeed, if network partitions occur, either the system remains available at the expense of consistency, or stops making progress until the network partition is resolved to guarantee consistency. 
This means that no blockchain can  simultaneously be available and consistent. However, by maintaining the candidate and the finalized chain simultaneously, Ethereum Proof-of-stake aims to offer both safety and availability. The candidate chain aims to be available but without guaranteeing consistency all the time, while the finalized chain falls on the other side of the spectrum, guaranteeing consistency without availability. Therefore, the finalized chain will finalize  blocks  only when it is safe to do so where  the candidate chain will still be available during network partitions (caused by network failures or attacks).   The only caveat here is that the finalized chain grows by finalizing blocks of the candidate chain, which means that the properties of the two chains are interdependent. In particular, to assure liveness, it is necessary that the candidate chain steadily grows. This interdependence is a source of  vulnerability that must be thoroughly analysed. 



\section{Ethereum Proof-of-Stake Protocol}
\label{sec:ethereumConsensusProtocol}

\subsection{Overview}
\label{subsec:overview}
The Ethereum Proof-of-Stake (PoS) protocol design is quite involved. We identify, similarly to \cite{neu_ebb_2021}, the objectives underlying its design as follows: (i) finalizing blocks and 
    (ii) having an available candidate chain that does not rely on block finality to grow. To this end, the Ethereum PoS protocol combines two blockchain designs: a Nakamoto-style protocol to build the tree of blocks containing the transactions and a BFT finalization protocol to progressively finalize blocks in the tree. The objective is to keep the blockchain creation process always available while guaranteeing the finalization of blocks through Byzantine-tolerant voting mechanisms. The finalization mechanism is a \emph{Finality Gadget} called \emph{Casper FFG}, and the fork choice rule to select candidate chains is \emph{LMD GHOST}.

Before introducing how the fork choice rule and the finality gadget work together, we will introduce the following basic concepts: (i) slots, epochs, and checkpoints, which set the pace of the protocol allowing validators to synchronize together on the different steps, (ii) committees formation and assignment of roles to validators as proposers and voters for each slot, and (iii) the different types of votes the validators must send in order to grow and maintain the candidate chain as well as the finalized chain.





In this section, we focus on providing a formal version of the protocol through pseudo-code, following the specification given by the Ethereum Foundation \cite{github_specs}. Every implementation of the protocol must be compliant with the specification. Note that a description of an initial plan of the protocol is proposed by Buterin et al. in \cite{buterin_combining_2020}. In this paper, we describe and formalize the current implementation of the protocol \cite{github_specs}.

\subsubsection{Slots, Epochs \& Checkpoints}
\label{subsubsec:time}
In proof-of-work protocols, such as originally in \cite{nakamoto_peer_2008}, the average frequency of the block creation is predetermined in the protocol, and the mining difficulty changes to follow that pace.
On the contrary, in Ethereum PoS, it is assumed that validators have synchronized clocks to propose blocks at regular intervals.
More specifically, in the protocol, time is measured in \emph{slots} and \emph{epochs}. A slot lasts 12 seconds.
Slots are assigned with consecutive numbers; the first slot is slot $0$. 
Slots are encapsulated in \emph{epochs}. An epoch is composed of 32 slots, thus lasting 6 minutes and 24 seconds. The first epoch (epoch 0) contains from slot 0 to slot 31; then epoch 1 contains slot 32 to 63, and so on. These slots and epochs allow associating the validators' roles to the corresponding time frame.
An essential feature of epochs is the \emph{checkpoint}. 
A checkpoint is a pair block-epoch $(b,e)$ where $b$ is the block of the first slot\footnote{In the event of an epoch without a block for the first slot, the block used for the checkpoint is the last block in the candidate chain, belonging then to a previous epoch. 
On the contrary, if the proposer of the first slot proposes multiple blocks, this will make multiple checkpoints for the other validators to choose from using the fork choice rule.
} of epoch $e$.

\subsubsection{Validators \& Committees}
\label{subsubsec:validator&committee}
Validators have two main roles: \emph{proposer} and \emph{attester}. The proposer's role consists in proposing a block during a specific slot\footnote{The current protocol specifications \cite{github_specs} indicate that correct validators should send their block proposition during the first third of their designated slot.}. 
This role is pseudo-randomly\footnote{Detailed explanation in \autoref{ssec:pseudo-randomness}.} given to 32 validators by epoch (one for each slot).
The attester's role consists in producing an attestation sharing the validator's view of the chain. This role is given once by epoch to each validator.

In each epoch, a validator is assigned to exactly one committee (of attesters). 
A committee $C_j$ is a subset of the whole set of validators. 
Each validator belongs to exactly one committee, i.e., $\forall j\neq k,  C_j \bigcap C_k = \emptyset $ and for each epoch $\bigcup_{i} C_i = \Pi $. Each committee is associated with a slot. During this slot, each member of the committee will have to cast an \emph{attestation} to indicate its view of the chain.

In short, during an epoch, validators are all attesters once and have a small probability of being proposers ($32/n$). The roles of proposer and attester are entirely distinct, i.e., the proposer of a slot is not necessarily an attester of that slot.

\subsubsection{Vote \& Attestation}
\label{subsubsec:vote}
There are two types of votes in Ethereum PoS, the \emph{block vote}\footnote{Also called GHOST vote in \cite{buterin_combining_2020} and in the specifications \cite{github_specs}.} and the \emph{checkpoint vote}\footnote{Also called FFG vote in \cite{buterin_combining_2020} and in the specifications \cite{github_specs}.}. The message containing these two votes is called \emph{attestation}. 
During an epoch, each validator must make one attestation. The attestation ought to be sent during a specific slot. This slot depends on the committee of which the validator is a member. 
The two types of votes,
checkpoint vote and block vote, have very distinct purposes. The checkpoint vote is used to finalize blocks to grow the finalized chain. The block vote is used to determine the candidate chain. 
Although validators cast their two types of votes in one attestation, an important distinction must be made between the two. Indeed, the two types of votes do not require the same condition to be taken into account. The checkpoint vote of an attestation is only considered when the attestation is included in a block. In contrast, the block vote is considered one slot after its emission, whether it is included in a block or not.
The code associated with the production of attestations is described in \autoref{algo:prepareAttestation} at \autoref{subsec:code}. We then describe in \autoref{algo:syncAttestation} how the reception of attestations is handled.

\subsubsection{Finality Gadget}
\label{subsubsec:finalityGadget}
The finality gadget is the mechanism that aims at finalizing blocks. The finality gadget grows the finalized chain disregarding the block production. This decoupling of the finality mechanism and the block production permits block availability even when the finalizing process is slowed down. This differs from protocols like Tendermint \cite{buchman_latest_2018}, where a new block can be added to the chain only after being finalized.


The finality gadget works at the level of epochs. Instead of finalizing blocks one by one, the protocol uses checkpoint votes to finalize entire epochs. 
We now present in more detail how the finality gadget of Ethereum PoS grows the finalized chain. 

Recall that to be taken into account, a checkpoint vote needs to be included in a block. The vote will then influence the behavior of validators regarding this particular branch. 
Thus, in \autoref{algo:Casper} of \autoref{subsec:code} the function \texttt{countMatchingCheckpointVote} only counts the matching checkpoint votes of attestations included in a block. 

\paragraph*{Justification} \label{paragraph:Justification}
The justification process is a step to achieve finalization\footnote{The genesis checkpoint (i.e., the checkpoint of the first epoch) is the exception to this rule: it is justified and finalized by definition.}. It operates on checkpoints at the level of epochs. 
Justification occurs thanks to checkpoint votes. The checkpoint vote contains a pair of checkpoints: the checkpoint \emph{source} and the checkpoint \emph{target}. We can count with \texttt{countMatchingCheckpointVote} the sum of balances of the validator's checkpoint votes with the same source and target. If validators controlling more than two-thirds of the stake make the same checkpoint vote, then we say there is a \emph{supermajority link} from the checkpoint source to the checkpoint target. The checkpoint target of a supermajority link is said to be \emph{justified}.

More formally, a checkpoint vote is in the form of a pair of checkpoints: $\big((a,e_a),(b,e_b)\big)$, also noted $(a,e_a) \xrightarrow{} (b,e_b)$. For the checkpoint vote $(a,e_a) \xrightarrow{} (b,e_b)$ we call $(a,e_a)$ the checkpoint source and $(b,e_b)$ the checkpoint target. The checkpoint source is necessarily of an earlier epoch than the checkpoint target, i.e., $e_a<e_b$. 
In line with \cite{buterin_combining_2020}, we say there is a \emph{supermajority link} from checkpoint $(a,e_a)$ to checkpoint $(b,e_b)$ if validators controlling more than two-thirds of the stake cast an attestation
with checkpoint vote $(a,e_a) \xrightarrow{} (b,e_b)$.
In this case, we write $(a,e_a) \xrightarrow{\texttt{J}} (b,e_b)$ and the checkpoint $(b,e_b)$ is justified.

\paragraph*{Finalization} \label{paragraph:Finalization}
     
The finalization process aims at finalizing checkpoints, thus growing the finalized chain. Checkpoints need to be justified before being finalized. Let us illustrate the finalization process with the two scenarios that can lead to finalization. The first case presents the main scenario in the synchronous setting. It shows how a checkpoint can be finalized in two epochs, the least amount of epochs needed for finalization.

\subparagraph*{Case 1:} The scenario is depicted in Figure \ref{fig:finalizationCase1}.
\begin{enumerate}
    \item Let $A=(a,e)$ and $B=(b,e+1)$ be checkpoints of two consecutive epochs such that $A=(a,e)$ is justified.
    \item A supermajority link occurs between checkpoints $A$ and $B$ where $A$ is the source and $B$ the target. This justifies checkpoint $B$.
    Hence, we can write: $(a, e) \xrightarrow{\texttt{J}} (b, e+1)$ or equivalently $A \xrightarrow{\texttt{J}} B$.
    \item This leads to $A$ being finalized.
\end{enumerate}

\begin{figure}[th]
    \centering
    
    \resizebox{.4\linewidth}{!}{
    \begin{tikzpicture}[
        finalizednode/.style={regular polygon, regular polygon sides=6,, draw=black, fill=green, minimum size=7mm},        finalizednode2/.style={regular polygon, regular polygon sides=6,, draw=black, minimum size=5.5mm},
        justifiednode/.style={regular polygon, regular polygon sides=6,, draw=black, minimum size=7mm},         
        justifiednode2/.style={regular polygon, regular polygon sides=6,, draw=black, minimum size=5.5mm},
        checkpointnode/.style={regular polygon, regular polygon sides=6,, draw=black, minimum size=7mm},
        squarednode/.style={regular polygon, regular polygon sides=6,, draw=gray!60, fill=gray!5, minimum size=5mm},
        ]
        \node (Dots1)      {$\cdots$};
        \node[justifiednode]      (A)       [right of= Dots1]              {$A$};
        \node[justifiednode2]      (A2)       [right of= Dots1]              {};
        \node[checkpointnode]        (B)       [right of= A] {$B$};
         \node (Dots2)   [right of= B]   {$\cdots$};
        
        \draw[-] (A.east) -- (B.west);
        \draw[-] (Dots1.east) -- (A.west);
        \draw[-] (B.east) -- (Dots2.west);
        
        \end{tikzpicture} 
        }
        \hspace{0cm}
        
    \resizebox{.4\linewidth}{!}{
    \begin{tikzpicture}[
        finalizednode/.style={regular polygon, regular polygon sides=6,, draw=black, fill=green, minimum size=7mm},        finalizednode2/.style={regular polygon, regular polygon sides=6,, draw=black, minimum size=5.5mm},
        justifiednode/.style={regular polygon, regular polygon sides=6,, draw=black, minimum size=7mm},         justifiednode2/.style={regular polygon, regular polygon sides=6,, draw=black, minimum size=5.5mm},
        checkpointnode/.style={regular polygon, regular polygon sides=6,, draw=black, minimum size=7mm},
        squarednode/.style={regular polygon, regular polygon sides=6,, draw=gray!60, fill=gray!5, minimum size=5mm},
        ]
        \node (Dots1)      {$\cdots$};
        \node[justifiednode]      (A)       [right of= Dots1]              {$A$};
        \node[justifiednode2]      (A2)       [right of= Dots1]              {};
        \node[justifiednode]        (B)       [right of= A] {$B$};
        \node[justifiednode2]        (B2)       [right of= A] {};
         \node (Dots2)   [right of= B]   {$\cdots$};
        
        \draw[-] (A.east) -- (B.west);
        \draw[-] (Dots1.east) -- (A.west);
        \draw[-] (B.east) -- (Dots2.west);
        \draw[->,double] (A.north) to[out=45]  (B.north);
        
        \end{tikzpicture}  
        }
        \hspace{0cm}
        
    \resizebox{.4\linewidth}{!}{
    \begin{tikzpicture}[
        finalizednode/.style={regular polygon, regular polygon sides=6,, draw=black, fill=green, minimum size=7mm},
        finalizednode2/.style={regular polygon, regular polygon sides=6,, draw=black, minimum size=5.5mm},
        justifiednode/.style={regular polygon, regular polygon sides=6,, draw=black, minimum size=7mm},         justifiednode2/.style={regular polygon, regular polygon sides=6,, draw=black, minimum size=5.5mm},
        checkpointnode/.style={regular polygon, regular polygon sides=6,, draw=black, minimum size=7mm},
        ]
        \node (Dots1)      {$\cdots$};
        \node[finalizednode]      (A)       [right of= Dots1]              {$A$};
        \node[finalizednode2]      (A2)       [right of= Dots1]              {};
        \node[justifiednode]        (B)       [right of= A] {$B$};
        \node[justifiednode2]        (B2)       [right of= A] {};
         \node (Dots2)   [right of= B]   {$\cdots$};
        
        \draw[-] (A.east) -- (B.west);
        \draw[-] (Dots1.east) -- (A.west);
        \draw[-] (B.east) -- (Dots2.west);
        \draw[->,double] (A.north) to[out=45]  (B.north);
        
        \end{tikzpicture} 
        }
    \caption{\small The figure depicts the finalization scenario of \textbf{Case 1} with the 3 steps from top to bottom. We represent a checkpoint with a hexagon,
    a justified checkpoint with a double hexagon, 
    and a finalized checkpoint with double hexagon coloured. 
    The arrow 
    between two checkpoints indicates a supermajority link.}
    \label{fig:finalizationCase1}
\end{figure}
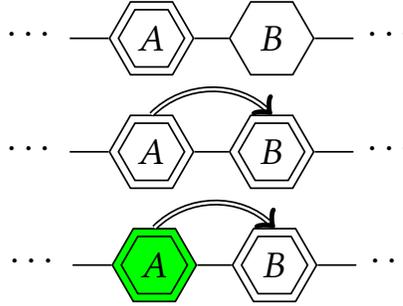

The second case illustrates the scenario in which two consecutive checkpoints are justified but not finalized. This means that the current highest justified checkpoint (e.g., $B$ in Figure \ref{fig:finalizationCase2}) was not justified with a supermajority link having the previous checkpoint $A$ as its source. Then occurs a new justification with the source and target being at the maximum distance (2 epochs) for the source to become finalized. An important note is that there is no limit on the distance between two checkpoints for justification to be possible. This limit only exists for finalization. 

\subparagraph*{Case 2: The scenario is depicted in \autoref{fig:finalizationCase2}.}
\begin{enumerate}
    \item Let $A=(a,e)$,  $B=(b,e+1)$ and $C=(c, e+2)$ be checkpoints of consecutive epochs such that $A$ and $B$ are justified. There is no supermajority link between $A$ and $B$, $A$ cannot be finalized as in Case 1 above.
    \item Now, a supermajority link occurs between checkpoints $A$ and $C$ where $A$ is the source and $C$ the target. This justifies checkpoint $C$, i.e., $A \xrightarrow{\texttt{J}} C$.
    \item This leads to $A$ being finalized. 
\end{enumerate}

\begin{figure}[th]
    \centering
    \resizebox{.5\linewidth}{!}{
    \begin{tikzpicture}[
    finalizednode/.style={regular polygon, regular polygon sides=6,, draw=black, fill=green, minimum size=7mm},        finalizednode2/.style={regular polygon, regular polygon sides=6,, draw=black, minimum size=5.5mm},
    justifiednode/.style={regular polygon, regular polygon sides=6,, draw=black, minimum size=7mm},         justifiednode2/.style={regular polygon, regular polygon sides=6,, draw=black, minimum size=5.5mm},
    checkpointnode/.style={regular polygon, regular polygon sides=6,, draw=black, minimum size=7mm},
    squarednode/.style={regular polygon, regular polygon sides=6,, draw=gray!60, fill=gray!5, minimum size=5mm},
    ]
    \node (Dots1)      {$\cdots$};
    \node[justifiednode]      (A)       [right of= Dots1]              {$A$};
    \node[justifiednode2]      (A2)       [right of= Dots1]              {};
    \node[justifiednode]        (B)       [right of= A] {$B$};
    \node[justifiednode2]        (B2)       [right of= A] {};
    \node[checkpointnode]        (C)       [right of= B] {$C$};
     \node (Dots2)   [right of= C]   {$\cdots$};
    
    \draw[-] (A.east) -- (B.west);
    \draw[-] (Dots1.east) -- (A.west);
    \draw[-] (B.east) -- (C.west);
    \draw[-] (C.east) -- (Dots2.west);
    
    \end{tikzpicture} 
    }
    \hspace{0cm}
    
    \resizebox{.5\linewidth}{!}{
    \begin{tikzpicture}[
    finalizednode/.style={regular polygon, regular polygon sides=6,, draw=black, fill=green, minimum size=7mm},        finalizednode2/.style={regular polygon, regular polygon sides=6,, draw=black, minimum size=5.5mm},
    justifiednode/.style={regular polygon, regular polygon sides=6,, draw=black, minimum size=7mm},         justifiednode2/.style={regular polygon, regular polygon sides=6,, draw=black, minimum size=5.5mm},
    checkpointnode/.style={regular polygon, regular polygon sides=6,, draw=black, minimum size=7mm},
    squarednode/.style={regular polygon, regular polygon sides=6,, draw=gray!60, fill=gray!5, minimum size=5mm},
    ]
    \node (Dots1)      {$\cdots$};
    \node[justifiednode]      (A)       [right of= Dots1]              {$A$};
    \node[justifiednode]        (B)       [right of= A] {$B$};
    \node[justifiednode]        (C)       [right of= B] {$C$};
    \node[justifiednode2]      (A2)       [right of= Dots1]              {};
    \node[justifiednode2]        (B2)       [right of= A] {};
    \node[justifiednode2]        (C2)       [right of= B] {};
     \node (Dots2)   [right of= C]   {$\cdots$};
    
    \draw[-] (A.east) -- (B.west);
    \draw[-] (Dots1.east) -- (A.west);
    \draw[-] (B.east) -- (C.west);
    \draw[-] (C.east) -- (Dots2.west);
    \draw[->,double] (A.north) to[out=45]  (C.north);
    \end{tikzpicture}
    }
    \hspace{0cm}
    
    \resizebox{.5\linewidth}{!}{
    \begin{tikzpicture}[
    finalizednode/.style={regular polygon, regular polygon sides=6,, draw=black, fill=green, minimum size=7mm},        finalizednode2/.style={regular polygon, regular polygon sides=6,, draw=black, minimum size=5.5mm},
    justifiednode/.style={regular polygon, regular polygon sides=6,, draw=black, minimum size=7mm},         justifiednode2/.style={regular polygon, regular polygon sides=6,, draw=black, minimum size=5.5mm},
    checkpointnode/.style={regular polygon, regular polygon sides=6,, draw=black, minimum size=7mm},
    squarednode/.style={regular polygon, regular polygon sides=6,, draw=gray!60, fill=gray!5, minimum size=5mm},
    ]
    \node (Dots1)      {$\cdots$};
    \node[finalizednode]      (A)       [right of= Dots1]              {$A$};
    \node[finalizednode2]      (A2)       [right of= Dots1]              {};
    \node[justifiednode]        (B)       [right of= A] {$B$};
    \node[justifiednode]        (C)       [right of= B] {$C$};
    \node[justifiednode2]        (B2)       [right of= A] {};
    \node[justifiednode2]        (C2)       [right of= B] {};
     \node (Dots2)   [right of= C]   {$\cdots$};
    
    \draw[-] (A.east) -- (B.west);
    \draw[-] (Dots1.east) -- (A.west);
    \draw[-] (B.east) -- (C.west);
    \draw[-] (C.east) -- (Dots2.west);
    \draw[->,double] (A.north) to[out=45]  (C.north);
    \end{tikzpicture}
    }
    \caption{\small The figure depicts the finalization scenario  of \textbf{case 2} with the 3 steps from top to bottom. 
    between two checkpoints indicates a supermajority link.}
    \label{fig:finalizationCase2}
\end{figure}
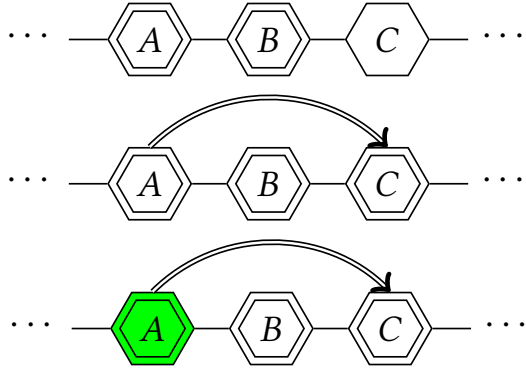

These two cases illustrate the fact that for a checkpoint to become finalized, it needs to be the source of a supermajority link between justified checkpoints. Once a checkpoint is finalized, all the blocks leading to it (including the block in the pair constituting the checkpoint) become finalized. 
We now describe the conditions for a checkpoint to be finalized more formally. Let $(a,e_a)$ and $(b,e_b)$ be two checkpoints such that $e_a<e_b$.
The checkpoint $(a,e_a)$ is finalized if the following conditions are respected: 
\begin{itemize}
    \item Source justified: checkpoint $(a,e_a)$ is justified.
    \item Supermajority link: there exists a supermajority link $(a,e_a) \xrightarrow{\texttt{J}} (b,e_b)$.
    \item Maximal gap: $e_b-e_a\leq 2$.\footnote{This last condition necessitating the two checkpoints to be at most 2 epochs away from each other is also called \emph{2-finality} \cite{buterin_combining_2020}.} Moreover, if $e_b-e_a=2$ then the checkpoint in between at epoch $e_a+1$($=e_b-1$) must be justified.
\end{itemize}
The importance of the last condition is illustrated by the \autoref{fig:notFinalizationCase}. 
In practice, these three conditions are only applied on the last four epochs. As mentioned in \cite{buterin_combining_2020}, at the implementation level, checkpoints more than 4 epochs old are not considered for finalization. All the conditions for finalization are illustrated by the last 4 conditions of \autoref{algo:Casper} in \autoref{subsec:code}. 

\begin{figure}[th]
    \centering
    \resizebox{.45\linewidth}{!}{
    \begin{tikzpicture}[
    finalizednode/.style={regular polygon, regular polygon sides=6,, draw=black, fill=green, minimum size=7mm},        finalizednode2/.style={regular polygon, regular polygon sides=6,, draw=black, minimum size=5.5mm},
    justifiednode/.style={regular polygon, regular polygon sides=6, draw=black, minimum size=7mm},
    justifiednode2/.style={regular polygon, regular polygon sides=6,, draw=black, minimum size=5.5mm},
    checkpointnode/.style={regular polygon, regular polygon sides=6, draw=black, minimum size=7mm},
    squarednode/.style={regular polygon, regular polygon sides=6, draw=gray!60, fill=gray!5, minimum size=5mm},
    ]
    \node (Dots1)      {$\cdots$};
    \node[justifiednode]      (A)       [right of= Dots1]              {$A$};
    \node[justifiednode2]      (A2)       [right of= Dots1]              {};
    \node[checkpointnode]        (B)       [right of= A] {$B$};
    \node[checkpointnode]        (C)       [right of= B] {$C$};
     \node (Dots2)   [right of= C]   {$\cdots$};
    
    \draw[-] (A.east) -- (B.west);
    \draw[-] (Dots1.east) -- (A.west);
    \draw[-] (B.east) -- (C.west);
    \draw[-] (C.east) -- (Dots2.west);
    
    \end{tikzpicture}   
    }
    \hspace{0cm}
    
    \resizebox{.45\linewidth}{!}{
    \begin{tikzpicture}[
    finalizednode/.style={regular polygon, regular polygon sides=6,, draw=black, fill=green, minimum size=7mm},        finalizednode2/.style={regular polygon, regular polygon sides=6,, draw=black, minimum size=5.5mm},
    justifiednode/.style={regular polygon, regular polygon sides=6,, draw=black, minimum size=7mm},         
    justifiednode2/.style={regular polygon, regular polygon sides=6,, draw=black, minimum size=5.5mm},
    checkpointnode/.style={regular polygon, regular polygon sides=6,, draw=black, minimum size=7mm},
    squarednode/.style={regular polygon, regular polygon sides=6,, draw=gray!60, fill=gray!5, minimum size=5mm},
    ]
    \node (Dots1)      {$\cdots$};
    \node[justifiednode]      (A)       [right of= Dots1]              {$A$};
    \node[justifiednode2]      (A2)       [right of= Dots1]              {};
    \node[checkpointnode]        (B)       [right of= A] {$B$};
    \node[justifiednode]        (C)       [right of= B] {$C$};
    \node[justifiednode2]        (C2)       [right of= B] {};
     \node (Dots2)   [right of= C]   {$\cdots$};
    \draw[-] (A.east) -- (B.west);
    \draw[-] (Dots1.east) -- (A.west);
    \draw[-] (B.east) -- (C.west);
    \draw[-] (C.east) -- (Dots2.west);
    \draw[->,double] (A.north) to[out=45]  (C.north);
    \end{tikzpicture}
    }
    \caption{\small This figure illustrates the case of two checkpoints $A$ and $C$ respecting all the conditions for finalization but the one that stipulates that a checkpoint $B$ in-between must be justified for $A$ to be finalized. 
    }
    \label{fig:notFinalizationCase}
\end{figure}
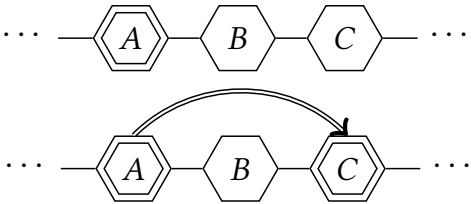

\subsubsection{Fork choice rule \& Block proposition}
\label{subsubsec:forkChoiceRule&blockProposition}
The fork choice rule is the mechanism that allows each validator to determine the candidate chain depending on their view of the BlockTree and the state of checkpoints. Ethereum PoS fork choice rule is LMD GHOST. The LMD GHOST fork choice rule stems from the Greedy Heaviest-Observed Sub-Tree (GHOST) rule \cite{sompolinsky_secure_2015}
which considers only each participant’s most recent vote (Latest Message Driven). 
During an epoch, each validator must make one \emph{block vote} on the block considered as the head of the candidate chain according to its view.
To determine the head of the candidate chain, the fork choice rule does the following:
\begin{enumerate}
    \item Go through the list of validators and check the last block vote of each.
    \item For each block vote, add a weight to each block of the chain that has the block voted as a descendent. The weight added is proportional to the stake of the corresponding validator. 
    \item Start from the block of the justified checkpoint with the highest epoch and continue the chain by following the block with the highest weight at each connection. Return the block without any child block. This block is the head of the candidate chain.
\end{enumerate}

The actual implementation is presented in \autoref{algo:GHOST} in \autoref{subsec:code}. This algorithm is similar to the one already presented in \cite{buterin_combining_2020}.
Albeit each \emph{block vote} being for a specific block, the fork choice rule considers all the chains leading to that block. This reflects the fact that a vote for a block is a vote for the chain leading to that block.  \autoref{fig:forkChoiceRuleExample} offers an explanation with a visualization of how attestations influence the fork choice rule. At each chain intersection, the fork choice rule favors the chain with the most attestations.

\begin{figure}[ht]
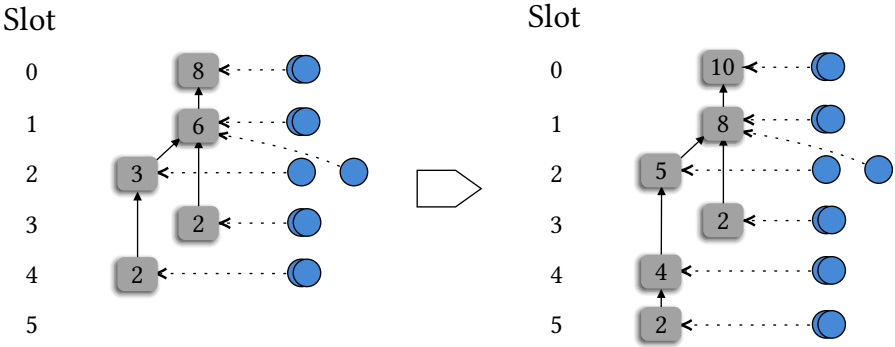

    \centering
    \resizebox{0.8\linewidth}{!}{

    }
 \caption{ \small Fork choice rule example observed from validator $i$'s point of view. We represent block votes with blue circles. 
    Block votes point to specific blocks indicating the block considered as the \emph{headblock} of the candidate chain at the moment of the vote.
    Each block has a number representing the value attributed by the fork choice rule algorithm (cf. \autoref{algo:GHOST}) to determine the candidate chain - we assume for this example that each validator has the same stake of 1. 
    On the left we represent the chain at the end of slot 4, and on the right at the end of slot 5.
    On the left, $i$'s fork choice rule gives the block of slot 4 as $C_i$'s head.
    On the right, the fork choice rule designates the block of slot 5 as the head of the candidate chain.}
    \label{fig:forkChoiceRuleExample}
\end{figure}

\subsubsection{Pseudo-Randomness}\label{ssec:pseudo-randomness}

The solution used by Ethereum to incorporate randomness in the consensus is called \textit{RANDAO}. RANDAO is a mechanism that creates pseudo-random numbers in a decentralized fashion. It works by aggregating different pseudo-random sources and mixing them.

\paragraph{Seed creation.} Each epoch produces a seed. This seed is created with the help of the block proposers of the said epoch. Each valid block contains a field called \texttt{randao\_reveal}\footnote{See \autoref{subsec:code} for the detailed explanation of its use.}. The seed is the hash of a \texttt{XOR} of all the \texttt{randao\_reveal} of an epoch plus an epoch number.

To prevent manipulation in the seed creation, each block's \texttt{randao\_reveal} must be the signature of specific data. 
The data to sign is the current epoch number. Anyone can then check that this signature is from the block proposer and for the correct data.

\paragraph{Seed utilization.} 
The algorithm using the seed is called \texttt{compute\_shuffled\_index} (cf. \autoref{algo:computeShuffledIndex}). This algorithm stems from the algorithm \textit{swap-or-not} that was introduced by \cite{hoang_enciphering_2012}.
\texttt{compute\_shuffled\_index} shuffles the validators list and gives new roles depending on their shuffled index. This pseudo-random shuffling function is used two times in Ethereum PoS consensus algorithm: for the proposer selection and for the committee selection. The proposer selection is described in \autoref{algo:getProposerIndex} and the committee selection in \autoref{algo:computeCommittee}.

\subsection{Pseudo Code}
\label{subsec:code}

In this section, we dive into a practical understanding of the mechanism behind the Ethereum PoS protocol. According to the specifications \cite{github_specs} and various implementations (such as Prysm \cite{prysm_code} and Teku \cite{teku_code}) we formalize the main functions of the protocol through pseudo-codes for a better understanding and for analysis purposes.

Each validator $p$ runs an instance of this particular pseudo code.
For instance, when a validator $p$ proposes a block, he broadcasts the following message: $\langle PROPOSE, ( slot,$ $\texttt{hash}(headBlock_p) ,$ $content) \rangle$, where $slot$ is the slot at which the proposer proposes the block, the hash of the $headBlock_p$ is the hash of the block considered to be the head of the candidate chain according to the fork choice rule (see \autoref{algo:GHOST}), and \textit{content} contains data used for pseudo-randomness, among other things that we will not detail here. We instead focus on the consensus protocol.

We describe in the following paragraphs the variables and functions used in the pseudo-code and the goal of these functions.

    \paragraph*{Variables.} During the computation, each variable takes a value that is subjective and may depend on the validator. We indicate with $p$ the fact that the value of variables depends on each process. The variable $tree_p$ is considered to be a graph of blocks with each block linked to its predecessor and thus represents the view of the blockchain (more precisely, $tree_p$ represents the view of all blocks received by the validator since the genesis of the system). 
Each $tree_p$ starts with the genesis block. 
$role_p$ corresponds to the different roles a validator can have, which is possibly none (i.e., for each slot, the validator can be proposer, attester, or have no role). $role_p$ is a list containing the role(s) of the validator for the current slot.
The $slot_p$ is a measure of time. In particular, a slot corresponds to 12 seconds. $slot_p \in \mathbb{N}$. Slot 0 begins at the time of the genesis block and is incremented every 12 seconds. $headblock_p$ is the head of the candidate chain according to $p$'s local view and the fork choice rule. 
A checkpoint $C$ is a pair block-epoch that is used for finalization. $C$ has two attributes which are $justified$ and $finalized$ which can be true or false (e.g., if $C$ is only justified $C.justified= $ \texttt{true} and $C.finalized= $ \texttt{false}). $lastJustifiedcheckpoint_p$ is the \textit{justified} checkpoint with the highest epoch. $currentCheckpoint_p$ is the checkpoint of the current epoch. The list $attestation_p$ is a list of size $n$ (i.e., the total number of validators). This list is updated only to contain the latest messages of validators (of at most one epoch old). 
$CheckpointVote_p$ is a pair of checkpoints, so a pair of pairs, used to make a checkpoint vote. Let us stress the fact that all those variables are local, and at any time, two different validators may have different valuations of those variables.

\paragraph*{Functions.} 
We describe the main functions of the protocol succinctly before giving the pseudo code associated and a more detailed explanation:
\begin{itemize}
    \item \texttt{validatorMain} is the primary function of the validator, which launches the execution of all subsidiary functions. 
    \item \texttt{sync} is a function that runs in parallel of the \texttt{validatorMain} function and ensures the synchronization of the validator. It updates the slot, the role(s) and processes justification and finalization at the end of the epoch and when a new validator joins the system.
    \item \texttt{getHeadBlock} applies the fork choice rule. This is the function that indicates the head block of the candidate chain.
    \item \texttt{justificationFinalization} is the function that handles the justification and finalization of checkpoints.
\end{itemize}

We depict in \autoref{algo:validatorMain} the main procedure of the validator. This procedure initializes all the values necessary to run a validator. In this paper, we consider the selection of validators already made to focus on the description of the consensus algorithm itself. The main starts a routine called \texttt{sync} to run in parallel. Then there is an infinite loop that handles the call to an appropriate function when a validator needs to take action for its role(s).

\begin{algorithm}[th]
\caption{Main code for a validator $p$}
\label{algo:validatorMain}
\begin{algorithmic}[1]
\setlength{\lineskip}{3pt}
\Procedure{validatorMain}{\ }
\State $tree_p \gets nil$ \Comment{The tree represents the linked received blocks} 
\State $role_p \gets [\ ]$ \Comment{$role_p $ can be ROLE\_PROPOSER and/or ROLE\_ATTESTER when it is not empty}
\State $slot_p \gets 0$ \Comment{$slot_p \in \mathbb{N}$}
\State $lastJustifiedCheckpoint_p \gets (0, genesisBlock)$ \Comment{A checkpoint is a tuple (epoch, block)}
\State $attestation_p \gets [\ ]$ \Comment{List of latest attestations received for each validator}
\State $validatorIndex_p \gets $ index of the validator \Comment{Each validator has a unique index} 
\State $listValidator \gets [p_0,p_1,\dots,p_{N-1} ]$ \Comment{A list of the validators index}
\State $balances \gets [\ ]$ \Comment{A list of the balances of the validators, their stake}
\State 
\State \textbf{start} \hyperref[algo:sync]{\texttt{sync}($tree_p,$ $slot_p,$  $attestation_p,$  $lastJustifiedCheckpoint_p,$ $role_p,$ 
%$listValidator$, 
$balances$)} 
\State
\While{\texttt{true}}
\If{$role_p\neq \emptyset$}
\If{ROLE\_PROPOSER $\in role_p$ } 
\State \hyperref[algo:prepareBlock]{\texttt{prepareBlock}}()
\EndIf
\If{ROLE\_ATTESTER  $\in role_p$} 
\State \hyperref[algo:prepareAttestation]{\texttt{prepareAttestation}}()
\EndIf
\State $role_p \gets [\ ] $
\Else
\State no role assigned \Comment{No action required}
\EndIf
\EndWhile
\EndProcedure
\end{algorithmic}
\end{algorithm}

The roles performed by the validator when proposer or attester are defined in \autoref{algo:prepareBlock} and \autoref{algo:prepareAttestation}, respectively. 
The proposer of a block does the three following tasks: 
\begin{enumerate}
    \item Get the head of its candidate chain to have a block to build upon;
    \item Sign a predefined pair to participate in the process of pseudo-randomness;
    \item Broadcast a new block built on top of the head of the candidate chain. 
\end{enumerate}

The attestation is composed of three parts: the slot, the block vote, and the checkpoint vote. The validator uses the fork choice rule presented in \autoref{algo:GHOST} to obtain the block chosen for the block vote. \autoref{algo:GHOST} and the one stemming from it, \autoref{algo:weight}, have already been defined in \cite{buterin_combining_2020}. We restated them here for the sake of completeness. For the checkpoint vote, an honest validator should always vote for the current epoch as the target and take the justified checkpoint with the highest epoch (i.e., \textit{lastJustifiedCheckpoint}) as the source. 
In order to broadcast this attestation, the attester must wait for one of two things, either a block has been proposed for this slot or 1/3 of the slot (i.e., 4 seconds) has elapsed. This is ensure by the function \texttt{waitForBlockOrOneThird}.

\begin{algorithm}[th]
\caption{broadcast block}
\label{algo:prepareBlock}
\begin{algorithmic}[1]
\setlength{\lineskip}{3pt}
\Procedure{prepareBlock}{ } 
    \State $headBlock_p \gets$ 
    \hyperref[algo:GHOST]{\texttt{getHeadBlock(    %$tree_p, attestation_p, latestJutifiedCheckpoint_p$
    )}} 
    \State $randaoReveal \gets $ \texttt{sign}( \texttt{epochOf($slot$)})
    \State \textbf{broadcast} $\langle PROPOSE, ( slot, \texttt{hash}(headBlock_p), randaoReveal_p,$ $ content) \rangle$
\EndProcedure
\end{algorithmic}
\end{algorithm}

\begin{algorithm}[th]
\caption{Broadcast Attestation}
\label{algo:prepareAttestation}
\begin{algorithmic}[1]
\setlength{\lineskip}{3pt}
\Procedure{prepareAttestation}{ } 
    \State \texttt{waitForBlockOrOneThird()}  \Comment{wait for a new block in this slot or $\frac{1}{3}$ of the slot
    }
    \State $headBlock_p \gets$ \hyperref[algo:GHOST]{\texttt{getHeadBlock(%$treep, attestation, latestJutifiedCheckpoint$
    )}}
    \State $currentCheckpoint_p \gets $ (first block of the epoch, \texttt{epochOf}($slot$) )
    \State $CheckpointVote_p \gets \big(lastJustifiedCheckpoint_p, currentCheckpoint_p \big) $
    \State \textbf{broadcast} $\langle ATTEST, (slot_p, \underbrace{\texttt{hash}(headBlock_p)}_\textrm{block vote}, \underbrace{CheckpointVote_p}_\textrm{checkpoint vote}) \rangle$
\EndProcedure
\end{algorithmic}
\end{algorithm}

The synchronization of the validator $p$ is handled by the function \texttt{sync} described in \autoref{algo:sync}. This algorithm allows the validator to update its view of the blockchain, in particular,  the current slot, the list of attestations, the last justified checkpoint, the validator's role, and the balances of all validators. To determine its role(s), the validator verify the index of the designated validator for the current slot and the set of indexes forming the committee of the current slot.
In more details, there are two conditions to give a role to a validator for the current slot. The first condition calls \autoref{algo:getProposerIndex} and gives the validator $p$ the role of proposer if its index is the one of the current proposer. The second condtion checks whether $p$ belongs to the committee of the current slot (see \autoref{algo:computeCommittee}). The two roles of proposer and attester are entirely distinct, i.e., the proposer of a slot is not necessarily an attester of that slot.

The synchronization function also starts two other routines which are \texttt{syncBlock} and \texttt{syncAttestation} corresponding to \autoref{algo:syncBlock} and \autoref{algo:syncAttestation}, respectively. 
These routines are used to handle the broadcast of proposers and attesters. In both functions, upon receiving a block or an attestation, the validator $p$ verifies that it is valid thanks to the \texttt{isValid} function. It is important to note that upon receiving a block, a validator can update the last justified checkpoint only if the current epoch has not started more than 8 slots ago. 
This particular condition is what the patch has brought to prevent a liveness attack (see \autoref{sec:livenessAttack}).

\begin{algorithm}[th]
\caption{Sync}
\label{algo:sync}
\begin{algorithmic}[1]
\setlength{\lineskip}{3pt}
\Procedure{sync}{$tree, slot, attestation,$ $role,$ $lastJustifiedCheckpoint,$} 
    \State \textbf{start} \hyperref[algo:syncBlock]{\texttt{syncBlock($slot, tree$)}}
    \State \textbf{start} \hyperref[algo:syncAttestation]{\texttt{syncAttestation($attestation$)}}
    \Repeat
        \State $previousSlot \gets slot$
        \State $slot \gets \lfloor$ time in seconds since genesis block / 12 $\rfloor$
        \If{$previousSlot \neq slot$} \Comment{If we start a new slot}
            \State \textit{roleSlotDone} $ \gets $ false
            \If{$validatorIndex_p$ = \hyperref[algo:getProposerIndex]{\texttt{getProposerIndex}}(\hyperref[algo:getSeed]{\texttt{getSeed}}(current epoch), $slot$) }
            \State \texttt{append} ROLE\_PROPOSER to $role_p$
            \EndIf
            \If{$validatorIndex_p \in$ \hyperref[algo:computeCommittee]{\texttt{computeCommittee}}(\hyperref[algo:getSeed]{\texttt{getSeed}}(current epoch), $slot$)}
            \State \texttt{append} ROLE\_ATTESTER to $role_p$
            \EndIf
        \EndIf
        \If{$slot \pmod{32} = 0$} \Comment{First slot of an epoch}
            \State \hyperref[algo:Casper]{\texttt{jutificationFinalization}($tree,$ $lastJustifiedCheckpoint$)}
        \EndIf
    \Until{{validator exit}}
\EndProcedure
\end{algorithmic}
\end{algorithm}

\begin{algorithm}[th]
\caption{Sync Block}
\label{algo:syncBlock}
\begin{algorithmic}[1]
\setlength{\lineskip}{3pt}
\Procedure{syncBlock}{$slot, tree$}
    \Upon{$\langle PROPOSE, ( slot_i, \texttt{hash}(headBlock_i) , randaoReveal_i,$ $ content_i) \rangle$}{validator $i$}
        \State $block \gets \langle PROPOSE, ( slot_i, \texttt{hash}(headBlock_i) , $ $ randaoReveal_i,$ $ content_i) \rangle$
        \If{\texttt{isValid}($block$)} 
            \State add $block$ to $tree$ 
        \If{$slot \pmod{32} \leq 8$ }
            \State update justified checkpoint if necessary
        \EndIf
        \EndIf
    \EndUpon
\EndProcedure
\end{algorithmic}
\end{algorithm}

\begin{algorithm}[th]
\caption{Sync Attestation}
\label{algo:syncAttestation}
\begin{algorithmic}[1]
\setlength{\lineskip}{3pt}
\Procedure{syncAttestation}{$attestation$} 
    \Upon{$\langle ATTEST, (slot_i, headBlock_i, checkpointEdge_i) \rangle$}{validator $i$}
        \State $attestation_i \gets \langle ATTEST, (slot_i, headBlock_i, checkpointEdge_i) \rangle$
        \If{\texttt{isValid}($attestation_i$)}
            \State $attestation[i] \gets attestation_i$ 
        \EndIf
    \EndUpon
\EndProcedure
\end{algorithmic}
\end{algorithm}
\begin{algorithm}[th]
\caption{Get Head Block}
\label{algo:GHOST}
\begin{algorithmic}[1]
\setlength{\lineskip}{3pt}
\Procedure{getHeadBlock}{ }
\State $block \gets$ block of the justified checkpoint with the highest epoch
\While{$block$ has at least one child}
\State $block \gets \underset{b'\text{ child of }block}{\argmax}$ \hyperref[algo:weight]{\texttt{weight}}($tree, Attestation, b'$)  
\State (ties are broken by hash of the block header)
\EndWhile
\State \Return $block$
\EndProcedure
\end{algorithmic}
\end{algorithm}

\begin{algorithm}[th]
\caption{Weight}
\label{algo:weight}
\begin{algorithmic}[1]
\setlength{\lineskip}{3pt}
\Procedure{weight}{$tree, Attestation,block$}
\State $w \gets 0$
\For{every validator $v_i$}
\If{$\exists a \in \ Attestation$ an attestation of $v_i$ for $block$ or a descendant of $block$}
\State $w \gets w + $ stake of $v_i$
\EndIf 
\EndFor
\State \Return $w$
\EndProcedure
\end{algorithmic}
\end{algorithm}

\autoref{algo:Casper} can be considered the most intricate one. This algorithm is responsible for the justification or finalization of the checkpoints at the end of an each epoch. To do so, it counts the number of checkpoint votes with the same source and target. If this number corresponds to more than 2/3 of the stake of all validators, then the target is considered justified for the validator running this algorithm. The last four conditions concern finalization. They verify among the last four checkpoints which one fulfills the conditions to become finalized. The conditions to become finalized are formally described in \autoref{subsec:overview} and can be summarized by: the checkpoint must be the source of a supermajority link, and all the checkpoints between the source and target included must be justified.

\begin{algorithm}[th]
\caption{Justification and Finalization}
\label{algo:Casper}
\begin{algorithmic}[1]
\setlength{\lineskip}{3pt}
\Procedure{jutificationFinalization}{$tree,$ $ lastJustifiedCheckpoint$}
\State $source \gets lastJustifiedCheckpoint$ 
\State $target \gets $ the current checkpoint
\State $nbCheckpointVote \gets $ \texttt{countMatchingCheckpointVote}($source,$ $target$)
\\
\hskip\algorithmicindent {\color{gray} $\triangleright$ \emph{justification process}:}
\If{ $nbCheckpointVote \geq \frac{2}{3} *$ total balance of validators }
\State $target.justified \gets $ \texttt{true}
\State $lastJustifiedCheckpoint \gets target$
\EndIf
\\
\hskip\algorithmicindent {\color{gray} $\triangleright$ \emph{finalization process}:}
\State $A,B,C,D \gets $ the last 4 checkpoints \Comment{With $D$ being the current checkpoint.}
\If{$A.justified$ $\land$ $B.justified$ $\land$ ($A \xrightarrow{\texttt{J}} C$)}
\State $A.finalized \gets $ \texttt{true} \Comment{Finalization of $A$}
\EndIf 
\If{$B.justified$ $\land$ ($B \xrightarrow{\texttt{J}} C$)}
\State $B.finalized \gets $ \texttt{true} \Comment{Finalization of $B$}
\EndIf 
\If{$B.justified$ $\land$ $C.justified$ $\land$ ($B \xrightarrow{\texttt{J}} D$)}
\State $B.finalized \gets $ \texttt{true} \Comment{Finalization of $B$}
\EndIf 
\If{$C.justified$ $\land$ ($C \xrightarrow{\texttt{J}} D$)} 
\State  $C.finalized \gets $ \texttt{true} \Comment{Finalization of $C$}
\EndIf 
\EndProcedure
\end{algorithmic}
\end{algorithm}


\begin{algorithm}
\caption{Get randao mix}
\label{algo:getRandaoMix}
\begin{algorithmic}[1]
\setlength{\lineskip}{3pt}
\Procedure{getRandaoMix}{$epoch$}
\State $mix \gets 0$
\State $headBlock \gets $ \hyperref[algo:GHOST]{\texttt{getHeadBlock}}()
\For{\textbf{each} $block$ parent of $headBlock$ and belonging to $epoch$}
\State $mix \gets mix \oplus \texttt{hash}(block.$randaoReveal)
\Comment{$\oplus$ is a bit-wise \texttt{XOR} operator}
\EndFor
\State \Return $mix$
\EndProcedure
\end{algorithmic}
\end{algorithm}

\begin{algorithm}
\caption{Get seed}
\label{algo:getSeed}
\begin{algorithmic}[1]
\setlength{\lineskip}{3pt}
\Procedure{getSeed}{$epoch$}
\State $mix \gets$ \hyperref[algo:getRandaoMix]{\texttt{getRandaoMix}}($epoch-2$) \Comment{The seed of an epoch $i$ is based on the randao mix of epoch $i-2$ }
\State \Return \texttt{hash}($epoch + mix$)
\EndProcedure
\end{algorithmic}
\end{algorithm}

\

The pseudo-randomness needs a different seed for each epoch to yield different results. This is ensured by hashing the RANDAO mix and the epoch number as shown in \autoref{algo:getSeed}. Adding the epoch number is helpful if no block is proposed during an entire epoch. This corner case would always result in the same seed if it were not for the epoch number.

The RANDAO mix is computed in \autoref{algo:getRandaoMix}. The computation of the RANDAO mix of a given epoch consists of \textit{XORing} all the \textit{randaoReveal} of the blocks in that particular epoch.  We consider only the blocks of that particular epoch that belong to the candidate chain. 

The RANDAO mix of epoch $e-2$ determines the role of validators in epoch $e$. Hence, with \autoref{algo:computeCommittee}, as soon as epoch $e-2$ is over, validators can know to which committee they belong to at epoch $e$. \texttt{computeCommittee} (\autoref{algo:computeCommittee}) is the function that, given a seed and an epoch, returns the list of validators index corresponding to the committee for the slot specified. The number of validators in each committee\footnote{In the actual implementation, committees have a maximum size of 2048 \cite{github_specs}.} is computed to be less than $N/32$ (with $N$ the number of validators). Then using the shuffled index computed with \autoref{algo:computeShuffledIndex} a committee of the given size is drawn according to the slot in argument.
All validators of the committee will have to perform the role of attester during this slot. 

Since the balance can change until the previous epoch, block proposers are known at the end of epoch $e-1$ for epoch $e$. \autoref{algo:getProposerIndex} is the one handling the selection of a proposer for a designated slot. It starts by creating a seed specifically for the slot in question. Then there is a loop starting with a pseudo-random selection of the validator's index. The loop stops only when a validator meets the condition criteria. This condition is equivalent to being selected with a probability depending on the balance. Thus, the validator with index \textit{proposerIndex} is selected with probability $\frac{\textit{effectiveBalance}}{32}$, with \textit{effectiveBalance} being the stake of \textit{proposerIndex} capped to $32$, i.e., $min(balance, 32)$. 
\

Both algorithms \ref{algo:getProposerIndex} and \ref{algo:computeCommittee} make the proposer and the committee selection resort to \autoref{algo:computeShuffledIndex} to imbue randomness. As mentioned in \autoref{sec:ethereumConsensusProtocol}, \autoref{algo:computeShuffledIndex} stems from the algorithm \textit{swap-or-not} \cite{hoang_enciphering_2012}. Its name helps us understand the principle behind the algorithm: select a validator and its opposite (based on a pivot) and swap them or not. The selection of the validator and the swap depend on the value of a hash. An essential aspect of this algorithm is that it can get the index of validators in the shuffled list without having to compute the shuffling of the whole list of validators. This reduces unnecessary computation.

\begin{algorithm}
\caption{Compute shufflex index}
\label{algo:computeShuffledIndex}
\begin{algorithmic}[1]
\setlength{\lineskip}{3pt}
\Procedure{computeShuffledIndex}{$index, seed, nbValidators$}
\For{$i = 0 $ \textbf{to} $90$}
\State $pivot \gets $ \texttt{hash}($seed + i$) $(mod\;$ \textit{nbValidators}$)$
\State $ flip \gets pivot + nbValidators - index \; (mod\;$ \textit{nbValidators}$)$
\State $position \gets \max(index,\; flip)$
\State $bit \gets $ \texttt{hash}($seed+i+position$)$\pmod{2}$
\If{$bit = 0 \; (mod\;$ \textit{nbValidators}$)$}
\State $index \gets flip$
\EndIf
\EndFor
\State \Return $index$
\EndProcedure
\end{algorithmic}
\end{algorithm}

\begin{algorithm}
\caption{Get proposer index}
\label{algo:getProposerIndex}
\begin{algorithmic}[1]
\setlength{\lineskip}{3pt}
\Procedure{getProposerIndex}{$seed$, $slot$}
\State MAX\_RANDOM\_BYTE $\gets 2^8-1$
\State $i \gets 0$
\State $proposerSeed \gets $ \texttt{hash}($seed$+$slot$)
\State $nbValidators \gets $ \texttt{length}($listValidator$)
\While{true}
\State $proposerIndex \gets  listValidator$\hyperref[algo:computeShuffledIndex]{[\texttt{computeShuffledIndex}($i, seed, nbValidators$)]}
\State $randomByte \gets$ first byte of \texttt{hash}($proposerSeed + i \pmod{nbValidators}$) 
\State \textit{effectiveBalance} $\gets listValidators[proposerIndex]$.effectiveBalance
\If{\textit{effectiveBalance} $*$ MAX\_RANDOM\_BYTE $\geq $ MAX\_EFFECTIVE\_BALANCE $* randomByte$  }
\State \Return $proposerIndex$
\EndIf
\State $i \gets i + 1$
\EndWhile
\EndProcedure
\end{algorithmic}
\end{algorithm}

\begin{algorithm}
\caption{Compute Committee}
\label{algo:computeCommittee}
\begin{algorithmic}[1]
\setlength{\lineskip}{3pt}
\Procedure{computeCommittee}{$seed, slot$}
\State $committee \gets [\ ]$
\State \textit{nbValidatorByCommittee} $\gets \lceil \texttt{lenght}(listValidator) / 32 \rceil$
\For{$i = (slot \pmod{32}) *$\textit{nbValidatorByCommittee}  \textbf{to} $(slot+1 \pmod{32}) *$\textit{nbValidatorByCommittee} $-1$ }
\State $committee$.\texttt{append}($listValidator$[\hyperref[algo:computeShuffledIndex]{\texttt{computeShuffledIndex}($i, seed, nbValidators$)}])
\EndFor
\State \Return $committee$
\EndProcedure
\end{algorithmic}
\end{algorithm}



\section{Liveness attack}
\label{sec:livenessAttack}

In this section, we describe a liveness attack called \emph{bouncing attack} that delays finality in a partially synchronous network after \texttt{GST}. 
Previous works also exhibit liveness attacks against the protocol using the intertwining of the fork choice rule and the finality gadget \cite{nakamura_analysis_2019, neu_ebb_2021}.  
To prevent this type of attack, the protocol now contains a ``patch'' \cite{pullRequest_bouncing_2022} that was suggested on the Ethereum research forum \cite{nakamura_prevention_2019}. We show that the implemented patch is insufficient and this type of attack is still possible if certain conditions are verified. 
This is a probabilistic liveness attack against the protocol of Ethereum Proof-of-Stake. 
Our attack can happen with less than 1/3 of Byzantine validators, as discussed in \autoref{subsec:probabilisticBouncingAttack}. 

\subsection{Bouncing Attack}
\label{subsec:bouncingAttack}

\begin{figure}[th]
    \centering
    
\begin{tikzpicture}[x=0.75pt,y=0.75pt,yscale=-1,xscale=1]

\draw  [draw opacity=0][fill={rgb, 255:red, 126; green, 211; blue, 33 }  ,fill opacity=1 ] (221.29,15.74) -- (221.29,24.86) -- (217.29,24.86) -- (217.29,15.74) -- cycle ;
\draw  [fill={rgb, 255:red, 132; green, 255; blue, 0 }  ,fill opacity=1 ] (52.64,96.08) -- (46.53,106.67) -- (34.31,106.67) -- (28.2,96.08) -- (34.31,85.5) -- (46.53,85.5) -- cycle ;
\draw   (50.05,96.27) -- (45.14,104.77) -- (35.32,104.77) -- (30.41,96.27) -- (35.32,87.76) -- (45.14,87.76) -- cycle ;
\draw [color={rgb, 255:red, 251; green, 193; blue, 96 }  ,draw opacity=1 ][line width=1.5]  [dash pattern={on 1.69pt off 2.76pt}]  (220,108) .. controls (220.47,103.35) and (220.07,107.78) .. (220.01,91.94) ;
\draw [shift={(220,88)}, rotate = 90] [fill={rgb, 255:red, 251; green, 193; blue, 96 }  ,fill opacity=1 ][line width=0.08]  [draw opacity=0] (10.92,-2.73) -- (0,0) -- (10.92,2.73) -- cycle    ;
\draw   (52.64,56.8) -- (46.53,67.38) -- (34.31,67.38) -- (28.2,56.8) -- (34.31,46.22) -- (46.53,46.22) -- cycle ;
\draw   (50.19,56.68) -- (45.28,65.19) -- (35.46,65.19) -- (30.56,56.68) -- (35.46,48.18) -- (45.28,48.18) -- cycle ;
\draw   (52.6,20.3) -- (46.49,30.88) -- (34.26,30.88) -- (28.15,20.3) -- (34.26,9.72) -- (46.49,9.72) -- cycle ;
\draw  [draw opacity=0][fill={rgb, 255:red, 208; green, 2; blue, 27 }  ,fill opacity=1 ] (222.92,52.24) -- (222.92,61.36) -- (218.92,61.36) -- (218.92,52.24) -- cycle ;
\draw  [draw opacity=0][fill={rgb, 255:red, 117; green, 214; blue, 11 }  ,fill opacity=1 ] (222.92,61.36) -- (221.98,61.36) -- (218.92,56.57) -- (218.92,54.25) -- (222.92,60.51) -- (222.92,61.36) -- cycle ;
\draw  [draw opacity=0][fill={rgb, 255:red, 117; green, 214; blue, 11 }  ,fill opacity=1 ] (218.92,61.36) -- (218.92,59.03) -- (220.41,61.36) -- (218.92,61.36) -- cycle ;
\draw  [draw opacity=0][fill={rgb, 255:red, 117; green, 214; blue, 11 }  ,fill opacity=1 ] (222.92,58.51) -- (218.92,52.24) -- (220.41,52.24) -- (222.92,56.19) -- (222.92,58.51) -- cycle ;
\draw  [draw opacity=0][fill={rgb, 255:red, 117; green, 214; blue, 11 }  ,fill opacity=1 ] (222.92,52.24) -- (222.92,54.3) -- (221.61,52.24) -- (222.92,52.24) -- cycle ;

\draw   (21.5,3) -- (431.5,3) -- (431.5,121) -- (21.5,121) -- cycle ;

\draw (68.42,20.3) node [anchor=west] [inner sep=0.75pt]   [align=left] {{\footnotesize checkpoint}};
\draw (68.42,56.8) node [anchor=west] [inner sep=0.75pt]   [align=left] {{\footnotesize justified checkpoint}};
\draw (68.42,96.08) node [anchor=west] [inner sep=0.75pt]   [align=left] {{\footnotesize finalizied checkpoint}};
\draw (231.42,20.3) node [anchor=west] [inner sep=0.75pt]  [font=\footnotesize] [align=left] {attestation from honest validateur};
\draw (232.92,56.8) node [anchor=west] [inner sep=0.75pt]  [font=\footnotesize] [align=left] {attestation from Byantine validateur};
\draw (234.92,98.77) node [anchor=west] [inner sep=0.75pt]  [font=\footnotesize] [align=left] {designate which checkpoint is the \\target checkpoint for an attestation};

\end{tikzpicture}

    \caption{\small This figure serves as a summary of the signification of the main diagrams of other figures.}
    \label{fig:finalizationCase1}
\end{figure}
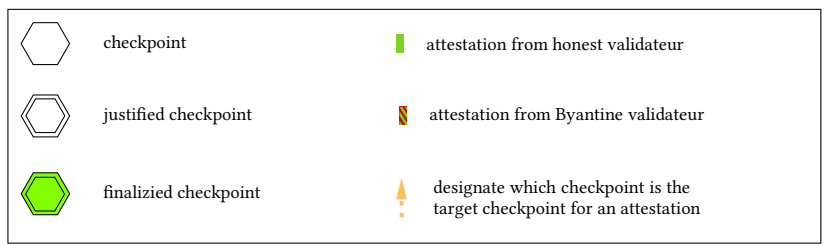

The \emph{Bouncing Attack} \cite{nakamura_analysis_2019} describes a liveness attack where the suffix of the chain changes repetitively between two candidate chains, thus preventing the chain from finalizing any checkpoint.
The Bouncing Attack exploits the fact that the candidate chains should start from the justified checkpoint with the highest epoch. 
It is possible for Byzantine validators to divide honest validators' opinions by justifying a new checkpoint once some honest validators have already cast their vote (made an attestation) during the asynchronous period before \texttt{GST}. 

The bouncing attack becomes possible once there is a \emph{justifiable} checkpoint in a different branch from the one designated by the fork choice rule with a higher epoch than the current highest justified checkpoint. 
A \emph{justifiable checkpoint} is a checkpoint that can become justified only by adding the checkpoint votes of Byzantine validators. 
If this setup occurs, the Byzantine validators could make honest validators start voting for a different checkpoint on a different chain, leaving a justifiable checkpoint again for them to repeat their attack and thus making validators \emph{bounce} between two different chains and not finalizing any checkpoint.
Hence the name Bouncing attack. 

Let us illustrate the attack with a concrete case. We show in \autoref{fig:simplifiedBoncingAttack}  an oversimplified case with only 10 validators, among which 3 are Byzantines.  
To occur, the attack needs to have a justifiable checkpoint with a higher epoch than the last justified checkpoint. We reach this situation before \texttt{GST}, which is presented in the left part of the figure. After reaching \texttt{GST}, Byzantine validators wait for honest validators to make a new checkpoint justifiable. When a new checkpoint is justifiable, the Byzantine validators cast their vote to justify another checkpoint, as shown by the right part of the figure. 
This will lead honest validators to start voting for the left branch, thus reaching a situation similar to the first step allowing the bouncing attack to continue. The repetition of this behavior is the bouncing attack. We emphasize this example in more detail in \autoref{fig:complexBouncingAttack} by detailing the sequence of votes allowing a ``bounce'' to occur and leaving a justifiable checkpoint on the other branch.

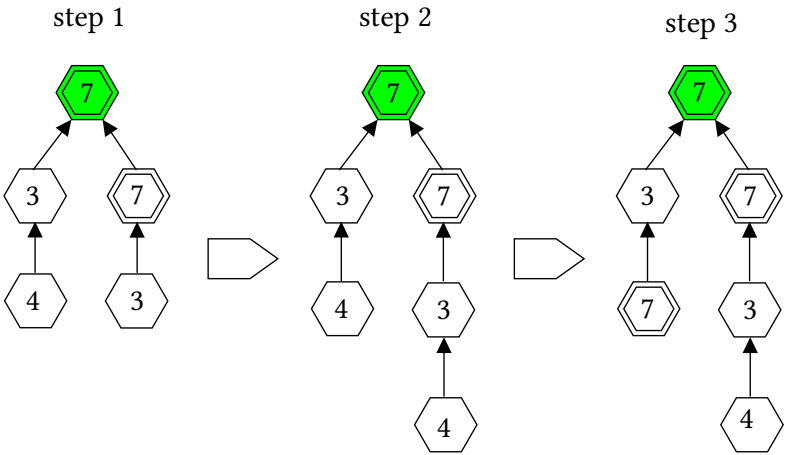
\begin{figure}
   \resizebox{.7\linewidth}{!}{
\begin{tikzpicture}[x=0.75pt,y=0.75pt,yscale=-1,xscale=1]

\draw    (22.18,109.35) -- (22.18,92.23) ;
\draw [shift={(22.18,89.23)}, rotate = 90] [fill={rgb, 255:red, 0; green, 0; blue, 0 }  ][line width=0.08]  [draw opacity=0] (6.25,-3) -- (0,0) -- (6.25,3) -- cycle    ;
\draw    (62.18,109.35) -- (62.18,92.29) ;
\draw [shift={(62.18,89.29)}, rotate = 90] [fill={rgb, 255:red, 0; green, 0; blue, 0 }  ][line width=0.08]  [draw opacity=0] (6.25,-3) -- (0,0) -- (6.25,3) -- cycle    ;
\draw   (34.59,120.32) -- (28.51,130.86) -- (16.34,130.86) -- (10.25,120.32) -- (16.34,109.78) -- (28.51,109.78) -- cycle ;
\draw   (74.19,120.37) -- (68.08,130.96) -- (55.86,130.96) -- (49.75,120.37) -- (55.86,109.79) -- (68.08,109.79) -- cycle ;
\draw   (34.44,79.18) -- (28.33,89.76) -- (16.11,89.76) -- (10,79.18) -- (16.11,68.59) -- (28.33,68.59) -- cycle ;
\draw   (74.94,79.18) -- (68.83,89.76) -- (56.61,89.76) -- (50.5,79.18) -- (56.61,68.59) -- (68.83,68.59) -- cycle ;
\draw  [fill={rgb, 255:red, 0; green, 255; blue, 0 }  ,fill opacity=1 ] (54.87,38.67) -- (48.76,49.25) -- (36.54,49.25) -- (30.42,38.67) -- (36.54,28.08) -- (48.76,28.08) -- cycle ;
\draw    (21.38,69.02) -- (34.61,51.64) ;
\draw [shift={(36.42,49.25)}, rotate = 127.27] [fill={rgb, 255:red, 0; green, 0; blue, 0 }  ][line width=0.08]  [draw opacity=0] (6.25,-3) -- (0,0) -- (6.25,3) -- cycle    ;
\draw    (61.88,69.02) -- (50.31,51.74) ;
\draw [shift={(48.65,49.25)}, rotate = 56.21] [fill={rgb, 255:red, 0; green, 0; blue, 0 }  ][line width=0.08]  [draw opacity=0] (6.25,-3) -- (0,0) -- (6.25,3) -- cycle    ;
\draw  [fill={rgb, 255:red, 255; green, 255; blue, 255 }  ,fill opacity=0 ] (52.45,38.67) -- (47.54,47.17) -- (37.72,47.17) -- (32.81,38.67) -- (37.72,30.16) -- (47.54,30.16) -- cycle ;
\draw  [fill={rgb, 255:red, 255; green, 255; blue, 255 }  ,fill opacity=1 ] (72.49,79.16) -- (67.58,87.66) -- (57.76,87.66) -- (52.85,79.16) -- (57.76,70.66) -- (67.58,70.66) -- cycle ;
\draw    (142.18,112.63) -- (142.18,92.23) ;
\draw [shift={(142.18,89.23)}, rotate = 90] [fill={rgb, 255:red, 0; green, 0; blue, 0 }  ][line width=0.08]  [draw opacity=0] (6.25,-3) -- (0,0) -- (6.25,3) -- cycle    ;
\draw    (182.18,112.63) -- (182.18,92.3) ;
\draw [shift={(182.18,89.3)}, rotate = 90] [fill={rgb, 255:red, 0; green, 0; blue, 0 }  ][line width=0.08]  [draw opacity=0] (6.25,-3) -- (0,0) -- (6.25,3) -- cycle    ;
\draw   (154.84,123.42) -- (148.76,133.96) -- (136.59,133.96) -- (130.5,123.42) -- (136.59,112.88) -- (148.76,112.88) -- cycle ;
\draw   (194.44,123.87) -- (188.33,134.46) -- (176.11,134.46) -- (170,123.87) -- (176.11,113.29) -- (188.33,113.29) -- cycle ;
\draw   (154.44,79.18) -- (148.33,89.76) -- (136.11,89.76) -- (130,79.18) -- (136.11,68.59) -- (148.33,68.59) -- cycle ;
\draw   (194.94,79.18) -- (188.83,89.76) -- (176.61,89.76) -- (170.5,79.18) -- (176.61,68.59) -- (188.83,68.59) -- cycle ;
\draw  [fill={rgb, 255:red, 0; green, 255; blue, 0 }  ,fill opacity=1 ] (174.87,38.67) -- (168.76,49.25) -- (156.54,49.25) -- (150.42,38.67) -- (156.54,28.08) -- (168.76,28.08) -- cycle ;
\draw    (141.38,69.02) -- (154.61,51.64) ;
\draw [shift={(156.42,49.25)}, rotate = 127.27] [fill={rgb, 255:red, 0; green, 0; blue, 0 }  ][line width=0.08]  [draw opacity=0] (6.25,-3) -- (0,0) -- (6.25,3) -- cycle    ;
\draw    (181.88,69.02) -- (170.31,51.74) ;
\draw [shift={(168.65,49.25)}, rotate = 56.21] [fill={rgb, 255:red, 0; green, 0; blue, 0 }  ][line width=0.08]  [draw opacity=0] (6.25,-3) -- (0,0) -- (6.25,3) -- cycle    ;
\draw  [fill={rgb, 255:red, 0; green, 0; blue, 0 }  ,fill opacity=0 ] (172.45,38.67) -- (167.54,47.17) -- (157.72,47.17) -- (152.81,38.67) -- (157.72,30.16) -- (167.54,30.16) -- cycle ;
\draw   (192.49,79.16) -- (187.58,87.66) -- (177.76,87.66) -- (172.85,79.16) -- (177.76,70.66) -- (187.58,70.66) -- cycle ;
\draw   (195.34,168.92) -- (189.23,179.51) -- (177.01,179.51) -- (170.9,168.92) -- (177.01,158.34) -- (189.23,158.34) -- cycle ;
\draw    (182.43,157.88) -- (182.43,137.55) ;
\draw [shift={(182.43,134.55)}, rotate = 90] [fill={rgb, 255:red, 0; green, 0; blue, 0 }  ][line width=0.08]  [draw opacity=0] (6.25,-3) -- (0,0) -- (6.25,3) -- cycle    ;
\draw    (262.18,112.63) -- (262.18,92.23) ;
\draw [shift={(262.18,89.23)}, rotate = 90] [fill={rgb, 255:red, 0; green, 0; blue, 0 }  ][line width=0.08]  [draw opacity=0] (6.25,-3) -- (0,0) -- (6.25,3) -- cycle    ;
\draw    (302.18,112.63) -- (302.18,92.3) ;
\draw [shift={(302.18,89.3)}, rotate = 90] [fill={rgb, 255:red, 0; green, 0; blue, 0 }  ][line width=0.08]  [draw opacity=0] (6.25,-3) -- (0,0) -- (6.25,3) -- cycle    ;
\draw   (274.84,123.42) -- (268.76,133.96) -- (256.59,133.96) -- (250.5,123.42) -- (256.59,112.88) -- (268.76,112.88) -- cycle ;
\draw   (314.44,123.87) -- (308.33,134.46) -- (296.11,134.46) -- (290,123.87) -- (296.11,113.29) -- (308.33,113.29) -- cycle ;
\draw   (274.44,79.18) -- (268.33,89.76) -- (256.11,89.76) -- (250,79.18) -- (256.11,68.59) -- (268.33,68.59) -- cycle ;
\draw   (314.94,79.18) -- (308.83,89.76) -- (296.61,89.76) -- (290.5,79.18) -- (296.61,68.59) -- (308.83,68.59) -- cycle ;
\draw  [fill={rgb, 255:red, 0; green, 255; blue, 0 }  ,fill opacity=1 ] (294.87,38.67) -- (288.76,49.25) -- (276.54,49.25) -- (270.42,38.67) -- (276.54,28.08) -- (288.76,28.08) -- cycle ;
\draw    (261.38,69.02) -- (274.61,51.64) ;
\draw [shift={(276.42,49.25)}, rotate = 127.27] [fill={rgb, 255:red, 0; green, 0; blue, 0 }  ][line width=0.08]  [draw opacity=0] (6.25,-3) -- (0,0) -- (6.25,3) -- cycle    ;
\draw    (301.88,69.02) -- (290.31,51.74) ;
\draw [shift={(288.65,49.25)}, rotate = 56.21] [fill={rgb, 255:red, 0; green, 0; blue, 0 }  ][line width=0.08]  [draw opacity=0] (6.25,-3) -- (0,0) -- (6.25,3) -- cycle    ;
\draw  [fill={rgb, 255:red, 0; green, 0; blue, 0 }  ,fill opacity=0 ] (292.45,38.67) -- (287.54,47.17) -- (277.72,47.17) -- (272.81,38.67) -- (277.72,30.16) -- (287.54,30.16) -- cycle ;
\draw   (312.49,79.16) -- (307.58,87.66) -- (297.76,87.66) -- (292.85,79.16) -- (297.76,70.66) -- (307.58,70.66) -- cycle ;
\draw   (315.34,168.92) -- (309.23,179.51) -- (297.01,179.51) -- (290.9,168.92) -- (297.01,158.34) -- (309.23,158.34) -- cycle ;
\draw   (272.51,123.47) -- (267.6,131.97) -- (257.78,131.97) -- (252.87,123.47) -- (257.78,114.97) -- (267.6,114.97) -- cycle ;
\draw    (302.43,157.88) -- (302.43,137.55) ;
\draw [shift={(302.43,134.55)}, rotate = 90] [fill={rgb, 255:red, 0; green, 0; blue, 0 }  ][line width=0.08]  [draw opacity=0] (6.25,-3) -- (0,0) -- (6.25,3) -- cycle    ;
\draw   (90,96) -- (106.4,96) -- (117.33,103.81) -- (106.4,111.62) -- (90,111.62) -- cycle ;
\draw   (210,96) -- (226.4,96) -- (237.33,103.81) -- (226.4,111.62) -- (210,111.62) -- cycle ;

\draw (62.18,78.83) node    {$7$};
\draw (21.38,120.35) node    {$4$};
\draw (62.18,120.45) node    {$3$};
\draw (42.65,38.7) node    {$7$};
\draw (21.38,79.2) node    {$3$};
\draw (182.43,79) node    {$7$};
\draw (142.18,123.47) node    {$4$};
\draw (182.43,123.95) node    {$3$};
\draw (162.63,38.7) node    {$7$};
\draw (142.85,79.2) node    {$3$};
\draw (182.43,170) node    {$4$};
\draw (302.43,79) node    {$7$};
\draw (262.18,123.47) node    {$7$};
\draw (302.43,123.95) node    {$3$};
\draw (282.65,38.2) node    {$7$};
\draw (262.18,79.2) node    {$3$};
\draw (296.93,161.7) node [anchor=north west][inner sep=0.75pt]    {$4$};
\draw (21.65,3.8) node [anchor=north west][inner sep=0.75pt]   [align=left] {\begin{minipage}[lt]{30.51pt}\setlength\topsep{0pt}
\begin{center}
step 1
\end{center}

\end{minipage}};
\draw (141.63,3.8) node [anchor=north west][inner sep=0.75pt]   [align=left] {\begin{minipage}[lt]{30.51pt}\setlength\topsep{0pt}
\begin{center}
step 2
\end{center}

\end{minipage}};
\draw (261.65,6.2) node [anchor=north west][inner sep=0.75pt]   [align=left] {\begin{minipage}[lt]{30.51pt}\setlength\topsep{0pt}
\begin{center}
step 3
\end{center}

\end{minipage}};

\end{tikzpicture}

}
    \caption{\small A bouncing attack presented in 3 steps. We have 10 validators, of which 3 are Byzantines. 
    The number inside each hexagon corresponds to the number of validators who made a checkpoint vote with this checkpoint as target.
    \textbf{1st step:} We start in a situation where there is a fork. A checkpoint is justified on one of the chains and a checkpoint of a higher epoch is \textit{justifiable} on the other. checkpoints are justified. We are at the end of the third epoch in which honest validators have divided their vote on each side. \textbf{2nd step:} We have reached \texttt{GST} at the beginning of the fourth epoch and 4 honest validators have already voted (rightfully so). \textbf{3rd step:} Here is the moment Byzantine validators take action and release their checkpoint vote for the concurrent chain, thus justifying the previously forsaken checkpoint and thereby changing the highest justifying checkpoint. By repeating this process, the bouncing attack can continue indefinitely.}
    \label{fig:simplifiedBoncingAttack}
\end{figure}

\begin{figure}
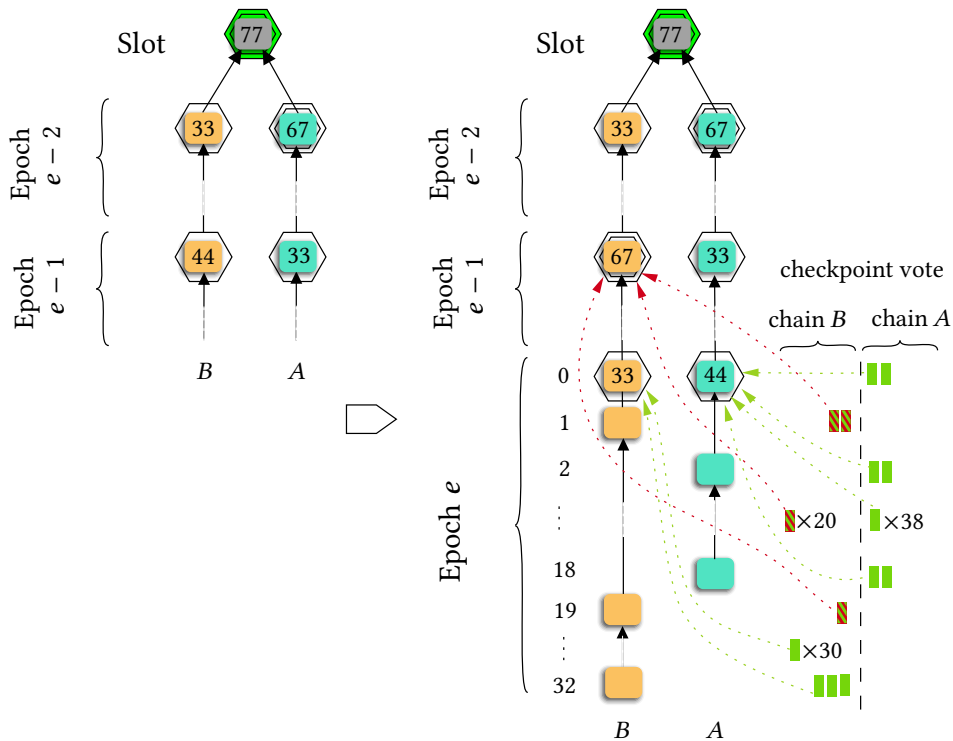

    \centering
    \resizebox{.85\linewidth}{!}{

}

    \caption{\small This figure presents a detailed version of the bouncing attack. In this example, we have a total of 100 validators, of which 23 are Byzantines.
    A block in a checkpoint corresponds to the block associated with that checkpoint.
    The number inside each hexagon (hovering a block) corresponds to the number of validators who made a checkpoint vote with this checkpoint as target. We distinguish between two sorts of checkpoint votes, the Byzantine ones, which are bi-color rectangles, 
    and the honest ones, which are uni-color rectangles. 
    We compile the 3 steps of \autoref{fig:simplifiedBoncingAttack} in 2 with more information on how justification's turning point is accomplished because of the Byzantine agents. \textbf{First step:} We begin from a situation where epoch $e-1$ just ended and we now reach \texttt{GST}. Notice that the candidate chain is chain $A$ because the checkpoint with the highest epoch is on chain $A$ but not chain $B$. \textbf{Second step:} In this step, the checkpoint vote released during epoch $e$ can change the last justified checkpoint to change the candidate chain for chain $A$ to chain $B$. Byzantine validators released their checkpoint vote from the previous epoch during epoch $e$. They send their last checkpoint vote at slot 23 once the checkpoint of epoch $e$ on chain $A$ has reached 44, thus becoming justifiable (i.e., not yet justified but with enough votes so that Byzantine validators can justify it). This trigger the candidate chain to change from chain $A$ to chain $B$ starting the \emph{bounce}.}
    \label{fig:complexBouncingAttack}
\end{figure}


\subsection{Implemented patch} \label{subsec:implementedPatch}
The explication of the patch is described for the first time on the Ethereum research forum \cite{nakamura_prevention_2019}.
The solution that was found to mitigate the bouncing attack is to engrave in the protocol the fact that validators cannot change their minds regarding justified checkpoints after a part of the epoch has passed. 

The goal of the solution proposed is to prevent the possibility of justifiable checkpoints being left out by honest validators.
To prevent honest validators from leaving a justifiable checkpoint, the patch must stop validators from changing their view of checkpoints before more than 1/3 of validators have cast their checkpoint vote.  
This condition stems from the fact that we reckon the proportion of Byzantine validators to be at most $1/3-\epsilon$. 
To apply this condition, the patch designates a number of slots after which honest validators cannot change their view of checkpoints.  
Indeed since validators are scattered equally among the different slots to cast their vote (in attestations) in a specific time frame, stopping validators from changing their view after a number of slots is equivalent to stopping changing their view after a certain proportion of validators have voted.
This does look like a solution to prevent Byzantine validators from influencing honest validators into forsaking a checkpoint now \emph{justifiable} for them. 

To enforce this behavior, called the ''fixation of view'', the protocol has a constant $j$ called \texttt{SAFE\_}\texttt{SLOTS\_}\texttt{TO}\texttt{\_UPDATE} \texttt{\_JUSTIFIED} in the code (cf. \autoref{algo:syncBlock} in \autoref{subsec:code}). 
This constant is the number of slots\footnote{At the time of writing this paper, $j=8$ \cite{github_specs}.} until when validators can change their view of the justified checkpoints. 
In the patch introducing this constant $j$, they mention a possible attack called \emph{splitting attack}. As they point out, the splitting attack relies on a ``last minute delivery" family of strategies whereby releasing a message late enough, some validators will consider it too late while others will not. This could thus split the validators into two different chains, not being able to conciliate their view before the end of the epoch. They consider the assumption of attackers being able to send a message at the right time to split honest validators too strong, we discuss in \autoref{sec:relatedWorks} whether it is the case or not. 
In \ref{subsec:probabilisticBouncingAttack}, we present a new attack inspired by the splitting attack with more realistic assumptions.

\subsection{Probabilistic Bouncing attack - why the patch is not enough}
\label{subsec:probabilisticBouncingAttack}

In this part, we present our novel attack against the protocol of Ethereum Proof-of-Stake. The attack is visually explained in \autoref{fig:probabilisticBouncingAttack}.

\begin{figure}
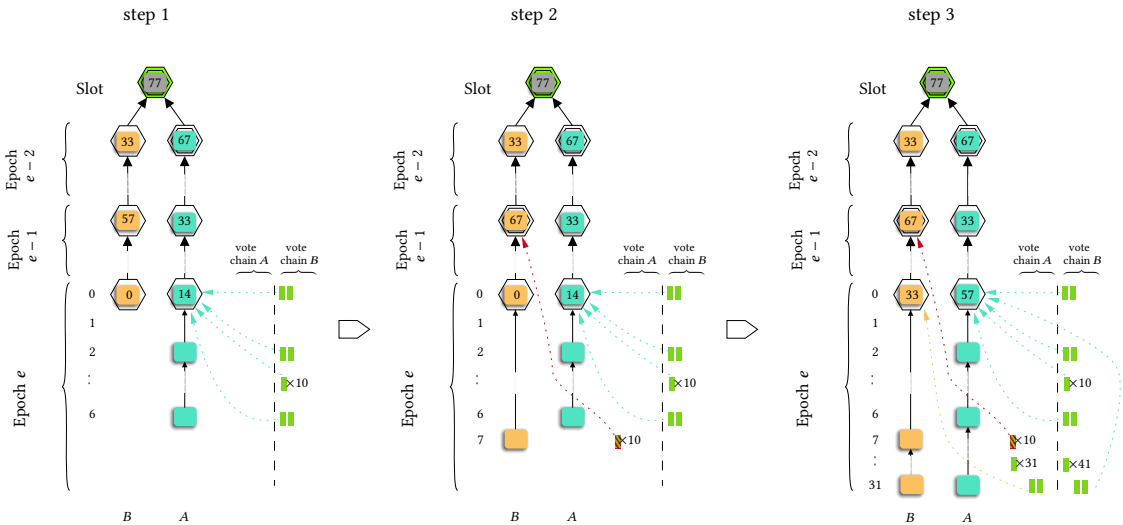

    \centering
    \resizebox{1\linewidth}{!}{



}
\caption{\small This figure presents the probabilistic bouncing attack. In this example, we consider 100 validators, of which 10 are Byzantines. A block in a checkpoint corresponds to the block associated with that checkpoint.
    The number inside each hexagon (hovering a block) corresponds to the number of validators who made a checkpoint vote with this checkpoint as target.
    The example starts at the slot before the attack in step 1. \texttt{GST} has been reach in epoch $e$ and honest validators have started to vote on chain $A$. This is the correct action because the justified checkpoint with the highest epoch is on chain $A$ (at epoch $e-2$).
    During the next slot in step 2, before reaching $limitSlot$, a Byzantine validator sends a block with withheld votes for the checkpoint at epoch $e-1$ on chain $B$.
    It is released just in time for a set of honest validators to consider it and too late for the remaining validators.
    The honest validators that see the block in time will update their view of the justified checkpoint with the highest epoch and consider chain $B$ as the candidate chain.
    We now show how the epoch continues with step 3. The block produced by a Byzantine having been released just in time,  $(1/3)$ of honest validators have changed their view. This results in a situation where Byzantine validators can perform the same attack during the next epoch provided that at least one Byzantine validator is selected to be block proposer on chain $A$ for one of the $8^{th}$ first slots.}
    \label{fig:probabilisticBouncingAttack}

\end{figure}

\subsubsection{Attack Conditions}
\label{subsubsec:attackConditions}

Our attack takes place during the synchronous period 
and uses the power of \textit{equivocation} of Byzantine processes.
Equivocation is caused by a Byzantine process that sends a message only to a subset of validators at a given point in time and potentially another message or none to another subset of validators. 
The effect is that only a part of the validators will receive the message on time. More in detail, the bounded network delay is used by a Byzantine validator to convey a message to be read on a specific slot by some validators and read on the next slot by the other validators. Note that if a protocol is not tolerant to equivocation, then it is not BFT (Byzantine-Fault Tolerant), since equivocation is the typical action possible for Byzantine validators.

\subsubsection{Attack Description and Analysis}

Let $\beta \le f/n$ be the fraction of Byzantine validators in the system. The attack setup is the following. First, as in the traditional bouncing attack, we start in a situation where the network is still partially synchronous. A fork occurs and results in the highest justified checkpoint being on chain $A$ at epoch $e$, and a justifiable checkpoint at epoch $e+1$ on chain $B$. Assume now that \texttt{GST} is reached, the attack can proceed\footnote{Note that before \texttt{GST}, no algorithm can ensure liveness since communication delays may not be bounded.} as follows:
\begin{enumerate}
    \item Since \texttt{GST} is reached, the network is fully synchronous. Chain $A$ is the candidate chain for all validators.
    \item Just before validators must stop updating their view concerning justified checkpoint (i.e., before reaching the limit of $j$ slots\footnote{At the time of writing, 8 slots.} in the epoch corresponding to the condition line 6 in \autoref{algo:syncBlock}), a Byzantine proposer proposes a block (cf. \autoref{algo:prepareBlock}) on chain $B$. 
    This block contains attestations with enough checkpoint votes to justify the justifiable checkpoint that was left by honest validators.
    The attestations included in the block are attestations of Byzantine validators that were not issued in the previous epoch when they were supposed to.
    The block must be released just in time, that is, right before the end of slot $j$, so that $(1/3-\beta)$ of the validators change their view of the candidate chain to be active on chain $B$ while the rest of honest validators continue on chain $A$. This is possible due to the patch preventing validators from changing their mind after $j$ slots.
    \item Repeat the process.
\end{enumerate}


An important aspect to consider in the attack is the probability for Byzantine validators to become proposers. This is an important part since without the role of proposer, validators are not legitimate to propose blocks and thus cannot add new attestations containing checkpoint votes on top of the concurrent chain.\footnote{Note that Byzantine validators can not use their role of proposer during the previous epoch to release a block with the right attestations because it might not be the last block of the epoch. Indeed because a part of honest validators is on the concurrent chain, they also add blocks to it. For the checkpoint votes contained in the Byzantine attestations to justify the justifiable checkpoint, it must be on the same chain as the attestation making the checkpoint justifiable in the first place.}. 
The probability of being selected to be a proposer directly impacts how long the probabilistic bouncing attack can continue. In the following theorem, we establish the probability of a probabilistic bouncing attack lasting for a specific number of epochs.

\begin{theorem}
The probabilistic bouncing attack occurs during $k$ epochs after \texttt{GST} and a favorable setting with probability:
\begin{equation}
    P(\textrm{bouncing }k \textrm{ times})=(1-\alpha^j)^k,
\end{equation}

with $\alpha \in [0,1]$ the proportion of honest validators and $j$ the number of slots before locking a choice for the justification.
\end{theorem}

\begin{proof}
We denote by $\alpha$ the proportion of honest validators and $j$ the number of slots before locking
the choice for justification.
We want to know the probability of delaying the finality for $k$ epochs. 
Once we assume a setup condition sufficient to start a probabilistic bouncing attack, the attack continues until it becomes impossible for Byzantine validators to cast a vote to justify the justifiable checkpoint. 
To cast their vote, Byzantine validators need one of the $j$ first slots of the concurrent 
chain to have a Byzantine validator as proposer. Considering the probability of choosing between each validator, the chance for a Byzantine validator to be a proposer for one of the first $j$ slots is $(1-\alpha^j)$; $\alpha$ being the proportion of honest validators. For $k$ epochs, we take this result to the power of $k$.
\end{proof}

As we can see, the probability of the bouncing attack to continue for $k$ epochs depends on two factors: $\alpha$ the proportions of honest validators which cannot be controlled and $j$ the number of slots before which validators are allowed to switch branches. 
Reducing $j$ to 0 would prevent the bouncing attack from happening (the probability falls to 0), but it would mean that validators are never allowed to change their view of the candidate chain. 
This naive solution would allow irreconcilable choices between the set of validators and prevent any new checkpoint from being justified, which is a more severe threat to the liveness of Ethereum PoS.

Reducing the number of slots where validators can change their view of the blockchain implies that different views cannot reconcile quickly. Or at least, the window of opportunity for doing so gets smaller. 
In theory, the proportion of Byzantine validators necessary to perform this attack is $1/n$. This is because we assume a favorable setup and that Byzantine validators can send messages so that only a wished portion of honest validators receives it on time.
Our analysis focuses on the course of action of the attackers during the attack rather than the course of action necessary for it to appear.


\section{Safety} \label{sec:safety}

In order to prove the safety of the protocol, we begin by giving lemmas concerning the justification of checkpoints. The first lemma rules out the possibility of two different justified checkpoints having the same epoch. New validators that wants to join the set of validators must send the amount they wish to stake at a specific smart contract\footnote{Currently 32 ETH is needed to become a validator.}. This transaction triggers the process for a validator to join the set of validators. The last step needed for the activation of a validator (allowing it to send attestations and propose blocks) requires that the block adding the validator to the validator set gets finalized\footnote{The exact process requires the placement of the validator in the activation queue to be finalized \href{https://github.com/ethereum/consensus-specs/blob/80ba16283c9447db8aa04eeaf4a3940b56480758/specs/phase0/beacon-chain.md}{https://github.com/ethereum/consensus-specs/blob/dev/specs/phase0/beacon-chain.md}.}.
This means that between two finalized checkpoints, the set of validators is fixed.

\begin{lemma} \label{firstLemma}
If checkpoints $C$ and $C'$ both of epoch $e$ are justified, it must necessarily be that $C = C'$.
\end{lemma}

\begin{proof}
By hypothesis, we know that Byzantine validators are at most $f<n/3$. For the sake of contradiction, let us assume that $C$ and $C'$ are different checkpoints.
Let $V$ be the set of at least $2n/3-f$ honest validators that cast a checkpoint vote for checkpoint $C$ in epoch $e$, and $V'$ be the set of at least $2n/3-f$ honest validators that cast a checkpoint vote for checkpoint $C'$ in epoch $e$. The intersection of the two sets of honest validators is $|V \cap V'| \geq (2n/3 - f) + (2n/3 - f) - (n - f) = (n/3 - f) > 0$. $|V \cap V'| > 0$ implies that at least one honest validator voted both for checkpoint $C$ and $C'$ in epoch e. This is a contradiction since, according to the protocol
specification\footnote{This is specified in the specs \url{https://github.com/ethereum/consensus-specs/blob/dev/specs/phase0/validator.md\#attester-slashing}, and implemented on actual client Prysm \url{https://github.com/prysmaticlabs/prysm/blob/0fd52539153e32cfbd0a27ee51f253f8f6bb71c4/validator/client/attest.go\#L140}. This corresponds to the only attestation done by an honest validator during an epoch, see \autoref{algo:sync}.}, 
an honest process signs at most one unique block per epoch, therefore $C=C'$. This proves there cannot be more than one justified checkpoint by epoch.
\end{proof}

The following lemma explains why the finalization of a checkpoint necessarily means that a checkpoint cannot be justified on a different chain afterward.

\begin{lemma}\label{secondLemma}
If a checkpoint $C$ of epoch $e$ is finalized on chain $c$, and a checkpoint $C'$ of epoch $e'$ is justified on chain $c'$ with $e'>e$, it necessarily means that $c$ and $c'$ have a common prefix until epoch $e$. 
\end{lemma}

\begin{proof}
$C'$ being justified on chain $c'$, it means that at least $2n/3-f$ honest validators must have cast a checkpoint vote with $C'$ as checkpoint target for epoch $e'$.

For the sake of contradiction, let us say that $c$ and $c'$ have a common prefix until epoch $e-1$ at most. For a checkpoint to be justified on chain $c'$ at an epoch strictly superior to $e$ it implies that a set $V'$ of at least $2n/3-f$ honest validators must have cast a checkpoint vote\footnote{See \autoref{paragraph:Justification} for details on justification.} with a checkpoint target on chain $c'$ and a checkpoint source with epoch inferior to $e-1$.

Checkpoint $C$ of epoch $e$ being finalized on chain $c$, we have two possibilities\footnote{See \autoref{paragraph:Finalization} for details on finalization.}. Either the checkpoint at epoch $e+1$ on chain $c$ has been justified with checkpoint $C$ as source. Or the checkpoint at epoch $e+2$ on chain $c$ has been justified with checkpoint $C$ as source and the checkpoint at epoch $e+1$ is justified. Either way, a justification occurred on chain $c$ with checkpoint $C$ as source, and no justification occurred on a different chain before its finalization.

Hence, we know that a set $V$ of at least $2n/3-f$ honest validators have cast a checkpoint vote with $C$ as checkpoint source before a justification on any other chain.

Seeing that $|V\cap V'|>0$, at least one honest validator has cast a checkpoint vote with $C$ as checkpoint source and then a checkpoint vote with a checkpoint source of at most epoch $e-1$ and a target on chain $c'$.

That means that at least one honest validator has cast a checkpoint vote with checkpoint source with epoch inferior to $e-1$ after seeing checkpoint $C$ at epoch $e$ justified. However, the fork choice rule of the protocol (cf. \autoref{algo:GHOST}) requires honest validators to vote on the chain with the highest justified checkpoint.
This contradiction proves the lemma. 
\end{proof}

We saw with Lemma \ref{firstLemma} that two checkpoints of the same epoch could not be justified, hence finalized. We then showed with Lemma \ref{secondLemma} that after a finalization on one chain, no checkpoints could become justified on any other chain. These are the conditions required to have safety, as we prove now.

\begin{theorem}[Safety]
There cannot be two finalized checkpoints on different chains.
\end{theorem}

\begin{proof}
Thanks to Lemma \ref{firstLemma}, we know that two different checkpoints $C$ and $C'$ of the same epoch cannot be justified, hence finalized. 

For the sake of contradiction, let us assume that two checkpoints $C$ and $C'$ are finalized on different chains $c$ and $c'$ at epoch $e$ and $e'$, respectively. We assume without loss of generality that $e<e'$. 
$C$ being finalized, we know thanks to Lemma \ref{secondLemma} that $C'$ cannot be justified on a different chain $c'$, let alone be finalized.
\end{proof}

The blockchain preserves the property of safety at all times. The Ethereum PoS is safe. 

\section{Related Works}
\label{sec:relatedWorks}

Blockchain protocol analyses can be divided into two main categories: those that specify or formalize protocols and those that identify vulnerabilities of the protocols. Our work is at the junction of the two categories because we
both formalize the protocol to permit its analysis, and we present
a novel attack. 

The category of specification and formalization includes the ``white papers''  (e.g., Bitcoin \cite{nakamoto_peer_2008}, Ethereum \cite{wood_ethereum_2014}). It also includes academic papers providing formal specifications and demonstrating  properties guaranteed by the protocols. For example, \cite{garay_bitcoin_2015} formally describes and analyses the Bitcoin protocol, proving its security guarantees, \cite{amoussou_dissecting_2019} does the same for the Tendermint's protocol (the consensus protocol of the Cosmos blockchain\cite{buchman_latest_2018}). Our work lies in this category of formalization by proposing a specification of  Ethereum Proof-of-Stake protocol's properties and a high-level description through pseudo-code. For what concerns protocol formalization of Ethereum Proof-of-Stake, \cite{buterin_combining_2020} is the first to propose draft specification of the Ethereum PoS protocol and related properties. However, that specification is outdated and not complete. Our paper provides a high-level formalization of the consensus mechanisms with respect to the current implemented code, along with a novel specification of its properties. 


A famous example of papers identifying vulnerabilities is \cite{eyal_majority_2018}, which presents the selfish mining attack on Bitcoin. \cite{eyal_majority_2018} shows that in Bitcoin (and proof-of-work in general), miners can benefit from deviating from the prescribed protocol by withholding blocks for a while at the expense of honest miners.
\cite{amoussou_correctness_2018} points out a liveness vulnerability of the Tendermint protocol. \cite{neuder_defending_2020} presents an attack where nodes can reorganize the Tezos' Emmy+ chain and then do a double-spend attack. 
Our work follows the line of research focusing on flaws  of Ethereum Proof-of-Stake  \cite{neu_ebb_2021,neu_two_2022,schwarz_three_2021}. Neu et al. \cite{neu_ebb_2021} exhibit a balancing attack, highlighting the shortcomings of a consensus mechanism being separated into two layers (finality gadget, fork choice rule). Mitigation against this attack was proposed, but \cite{neu_two_2022} overcame this mitigation with a new balancing attack. \cite{schwarz_three_2021} presented reorg attacks revealing that validators with the role of proposer could gain from disturbing the protocol by releasing their block late. 
Our paper presents another flaw regarding the liveness of the current Ethereum Proof-of-stake protocol, on which some attacks do not seem feasible anymore, and thus insists on the importance of finding new ways to conciliate availability and finality. We differ from \cite{galletta_resilience_2022} that aims at formally verifying the protocol of \textit{Hybrid Capser} focusing on an outdated version of the protocol. While we formalize through pseudo-code the current implemented version of the protocol and exhibit a liveness attack on the protocol.

Nakamura \cite{nakamura_prevention_2019}  presents an attack called \textit{splitting attack} in which the adversary sends messages to split the set of validators. However, to achieve this attack, Nakamura assumes that the adversary needs to control and play with network delays. This is a strong assumption and can be considered unrealistic. More recently, \cite{schwarz_three_2021} showed through experiments that attackers can predict the proportion of validators receiving a given message within a specific time frame with sufficient accuracy. This contradicts Nakamura's claim that the attack necessitates the adversary to control the network delay. 
In this work, we present a form of the splitting attack based on the weaker assumption that the adversary knows the network delay (in line with \cite{schwarz_three_2021}) but does not control it. Moreover, our attack is repeated, hence the name bouncing, being a threat to the liveness.

Outside of these two categories lies works that provide formal ground for blockchains. \cite{anceaume_finality_2021} describes the types of finality a blockchain can achieve, \cite{anceaume_abstract_2018} proposes a formalization  of blockchains and their evolutions as BlockTrees. We rely on the definition of BlockTree and finality to express the Ethereum protocol properties.

\section{Categorization of Attacks on Ethereum Proof-of-Stake}
\label{sec:categorizationAttacks}

Ethereum Proof-of-Stake has a history of attacks targeting the liveness of the protocol. These attacks can be divided into two main groups: the one targeting the fork choice rule -\emph{Partitioning over the candidate chains}-, and the one targeting the justification of checkpoints -\emph{Swinging over the justified chains}. The former aims at dividing the set of honest validators equally on two concurrent chains. The latter aims at preventing finalization by alternatively justifying on two concurrent chains.

This categorization is summarized in \autoref{table:attacksCategorization}. The controllable setup indicates the ease for Byzantine validators to start the attack. A controllable setup is one that Byzantines can perform during the synchronous period without requiring a prior state in the asynchronous period. The Resulting Effect expresses the state of the blockchain the attack achieves on  Liveness. The Detectable columun indicates if honest validators can detect Byzantine validators that mounted the attack. 

\paragraph{Partitioning over the candidate chains.} The first attack in this category is the Balancing attack \cite{neu_ebb_2021}. As for all the Partition attacks, the Balancing attack's purpose is to keep half of the honest validators on one chain and half on another. This fork is then balanced using the Byzantines' votes to ensure that honest validators are kept in check. The execution of the attack is the following. The first proposer of the epoch needs to be Byzantine and proposes 2 blocks. The 2 blocks are released to different parts of the networks to ensure that honest validators disagree on which block was first and thus the one belonging to the candidate chain. Byzantine validators then use their vote to keep the number of validators on each chain balanced.
This attack is detectable since the first proposer needs to propose 2 different blocks. It is the only one among the Byzantine validators that performs a visibly reproachable action.

Conversely to other attack that aims at partitioning the validators over two candidates, Refined Balancing attack \cite{schwarz_three_2021} don't need a Byzantine validator to perform a detectable action. To launch the attack, the Byzantine validators wait for an opportune epoch in which they are elected to propose a block for the two first slots.  
The Refined Balancing attack is a Balancing attack without the assumption of adversarial network delay. In their model, the adversarial network delay is the capacity for an adversary to arbitrarily delay messages sent in the network, bounded by the network delay. They manage to perform a Balancing attack without an adversarial network delay by gaining knowledge about the network delay by having Byzantine validators scattered in the network. To mitigate this attack and the previous Balance attack, a 'proposer boost score' was suggested and implemented\footnote{The pull request to add this mitigation can be found here: \href{https://github.com/ethereum/consensus-specs/pull/2730}{https://github.com/ethereum/consensus-specs/pull/2730}.}. This mitigation did not actually stop the attack as shown by the LMD-Specific attack.

Another Partition attack is the LMD-Specific Balancing attack \cite{neu_two_2022}. This attack works similarly to other balancing attacks; Byzantine validators create two concurrent chains, and they keep the two chains balanced (half honest validator on each). The main difference with other balancing attacks is that instead of only having a noticeable Byzantine block proposer, a handful of Byzantine validators release votes for both chains using equivocation. This creates two sets of honest validators with different candidate chains depending on which vote they have seen first. This attack is effective even after the patch proposed to stop the Balancing Attack.

\paragraph{Swinging over the justified chains.} This category of attacks starts with the Bouncing attack \cite{nakamura_analysis_2019}. This attack first outlines a limitation of the Casper FFG protocol, the underlying protocol used to finalize blocks. This limitation resides in that honest validators have to cast checkpoint votes for checkpoints stemming from the last justified checkpoint (having the last justified checkpoint as an ancestor). Using this restriction, Byzantine validators can split the set of validators by ensuring that they don't have the same view of the last justified checkpoint when they make their checkpoint vote. This restriction to only vote for checkpoints (and block) stemming from the justified checkpoints with the highest epoch set by the fork choice rule is used for all justification attacks. We presented the patch (cf. \autoref{subsec:implementedPatch}) that has been implement to prevent this attack.

The splitting attack \cite{nakamura_prevention_2019} presents an attack in which Byzantine validators with 'strong control over the network' can split honest validators into two sets. They can do so by sending attestations late enough for some validators to receive them on time while the rest don't. This attack can lead to a situation where a checkpoint gets justifiable but not justified. This means that the set of validators is divided between two checkpoints they try justify.
The attack's type is thus a partition over the justified chains.

Our attack differentiates from the splitting attack from two main points. We do not assume Byzantine validators to have control over the network but rather they possess the usual power of equivocation. We also repeat the process of dividing the set of validators on two different chains at each epoch while expressing the probability of the attack to continue. This makes this attack's type fit the swinging on justify chains.

All the attacks to justification chains do not have a controllable setup and share the feature of not being detectable. 
 These attacks rely mainly on messages being viewed at a certain point in time for some set of honest validators and at another time for the rest. There is no need for Byzantine validators to duplicate votes; they just have to withhold their messages/votes and release at the appropriate time.

\begin{table}
\caption{Summary of attacks against Ethereum's liveness}
\label{table:attacksCategorization}
\centering
\resizebox{0.9\linewidth}{!}{%
\begin{tabular}{|>{\hspace{0pt}}m{0.2\linewidth}|c|>{\hspace{0pt}}m{0.206\linewidth}|>{\hspace{0pt}}m{0.221\linewidth}|c|}
\hline
Attack original name & Controllable Setup & Attack's type & Resulting liveness & Detectable \\

\hline

"Balancing Attack" \cite{neu_ebb_2021} & \cmark & Partitioning over the candidate chains & Stopped & \cmark \\

\hline

"Refined Balancing Attack" (Balancing attack without adversarial network delay) \cite{schwarz_three_2021} & \cmark & Partitioning over candidate chains  & Stopped & \xmark \\

\hline

"LMD-Specific Balancing Attack" \cite{neu_two_2022} & \cmark & Partitioning over candidate chains & Stopped & \cmark \\

\hline

"Bouncing Attack" \cite{nakamura_analysis_2019} & \xmark & Swinging on justified chains & Stopped & \xmark \\

\hline

"Splitting Attack" \cite{nakamura_prevention_2019} & \xmark & Partitioning over the justified chains  & Delayed for 1 epoch & \xmark \\

\hline

Probabilistic Bouncing Attack (this work) & \xmark & Swinging on justified chains & Delayed for $k$ epochs with decreasing probability & \xmark \\

\hline

\end{tabular}
}
\end{table}

\section{Conclusion}
\label{sec:conclusion}
We described a framework for 
a high-level description of the protocol.
This is the first step in providing a formalization for verification tools.
We proposed a novel distinction between the definition of liveness and availability. This distinction is crucial to pinpoint the difference between Nakamoto-style and BFT consensus. It makes possible a comparison between the two. We outlined an attack against the liveness of the protocol, showing probabilistic liveness of the protocol under this attack and we prove safety of the protocol. 
Another aspect is that our analysis did not consider the protocol's rewards and incentives. We leave this part, as well as analysing rational behavior in the protocol, as future work. 

\bibliographystyle{ACM-Reference-Format}
\bibliography{main} 

\end{document}